\documentclass[11pt]{article}

\usepackage{times}

\usepackage{amsmath}
\usepackage{amsthm}
\usepackage{amssymb}
\usepackage{graphicx}
\usepackage{color} 
\usepackage{setspace} 
\usepackage{rotating}
\usepackage{natbib}

\usepackage{multirow}
\usepackage{xspace}
\usepackage{lscape}
\usepackage{xr}
\usepackage{hyperref}
\hypersetup{
  colorlinks,
  urlcolor=blue,
  linkcolor=black,
  citecolor=black
}

\usepackage[labelfont=bf,labelsep=period,justification=raggedright]{caption}

\makeatletter
\renewcommand{\@biblabel}[1]{\quad#1.}

\newtheorem{theorem}{Theorem}
\newtheorem{lemma}[theorem]{Lemma}

\newcommand{\vect}[1]{\boldsymbol{\mathbf{#1}}}
\newcommand{\ld}{\mathcal{L}}
\newcommand{\comment}[1]{}

\newcommand{\ins}{\textsc{insight}\xspace}

\newcommand{\sel}{\text{sel}}
\newcommand{\neut}{\text{neut}}
\newcommand{\data}{\vect{X}}
\newcommand{\outgroup}{\vect{O}}

\newcommand{\pmaj}{p(A_i=X_i^{maj})}
\newcommand{\pmin}{p(A_i=X_i^{min})}

\newcommand{\psd}{p(S_i=\text{\sel},Z_i\neq X_i^{maj})}
\newcommand{\psnd}{p(S_i=\text{\sel},Z_i=X_i^{maj})}
\newcommand{\pn}{p(S_i=\text{\neut})}

\newcommand{\PD}{D_\text{p}}
\newcommand{\WP}{P_\text{w}}

\newcommand{\Ea}{{\mathbb{E}[\PD]}}

\newcommand{\Ew}{{\mathbb{E}[\WP]}}

\def\@cite#1#2{(#1\if@tempswa , #2\fi)}
\def\@biblabel#1{}
\makeatother

\begin{document}

\begin{titlepage}

\title{Inference of Natural Selection from Interspersed Genomic Elements
Based on Polymorphism and Divergence}

\author{Ilan Gronau$^{1}$, Leonardo Arbiza$^{1}$,  Jaaved Mohammed$^{1,2}$, and Adam Siepel$^{1}$}

\date{ }
\maketitle

\begin{footnotesize}
\begin{center}
$^1$Department of Biological Statistics and Computational Biology,
Cornell University, Ithaca, NY 14853, USA 
\\[1ex]
$^2$Tri-Institutional Training Program in Computational Biology and Medicine 
\end{center}
\end{footnotesize}

\vspace{1in}

\begin{tabular}{lp{4.5in}}
{\bf Submission type:}& Research Article
\vspace{1ex}\\
{\bf Keywords:}& 
Molecular evolution, population genetics, noncoding DNA, regulatory
sequences, probabilistic graphical models 
\vspace{1ex}\\
{\bf Running Head:}&Inference of Selection from Interspersed Genomic Elements
\vspace{1ex}\\ 
{\bf Corresponding Author:}&
\begin{minipage}[t]{4in}
 Adam Siepel\\
 102E Weill Hall, Cornell University\\
 Ithaca, NY 14853\\
 Phone: +1-607-254-1157\\
 Fax: +1-607-255-4698\\
 Email: acs4@cornell.edu
\end{minipage}\\
\end{tabular}
\vspace{7ex}\\
{\bf This is an electronic version of an article published in {\em Mol Biol Evol}, 2013. \href{http://mbe.oxfordjournals.org/content/30/5/1159}{doi:10.1093/molbev/mst019}.}
 
\thispagestyle{empty}
\end{titlepage}

\doublespacing

\section*{Abstract}



Complete genome sequences contain valuable information about natural
selection, but extracting this information for short, widely scattered
noncoding elements remains a challenging problem.  Here we introduce a new
computational method for addressing this problem called Inference of
Natural Selection from Interspersed Genomically coHerent elemenTs (\ins).
\ins uses a generative probabilistic model to contrast patterns of
polymorphism and divergence in the elements of interest with those in
flanking neutral sites, pooling weak information from many short elements
in a manner that accounts for variation among loci in mutation rates and
genealogical backgrounds.  The method is able to disentangle the
contributions of weak negative, strong negative, and positive selection
based on their distinct effects on patterns of polymorphism and divergence.
Information about divergence is obtained from multiple outgroup genomes
using a full phylogenetic model.  The model is efficiently fitted to
genome-wide data by decomposing the maximum likelihood estimation procedure
into three straightforward stages.  The key selection-related parameters
are estimated by expectation maximization.  Using simulations, we show that
\ins can accurately estimate several parameters of interest even in complex
demographic scenarios.  We apply our methods to noncoding RNAs, promoter
regions, and transcription factor binding sites in the human genome, and
find clear evidence of natural selection.  We also present a detailed
analysis of particular nucleotide positions within GATA2 binding sites and
primary micro-RNA transcripts.


\thispagestyle{empty}
\clearpage
\setcounter{page}{1}

\section*{Introduction}


Evolutionary modeling has become an essential tool in genomic analysis.  It
is particularly valuable in the study of noncoding genomic elements in
large, complex eukaryotic genomes, because these elements are often
sparsely annotated, poorly understood, and difficult to examine
experimentally.  Rapid growth in the availability of complete genome
sequences, both within and across species, has led to many new
opportunities for evolutionary genomic analysis.  Among other things,
evolutionary models have been used to measure the
fractions of nucleotides likely to have fitness-influencing functions
\citep{CONS02,CHIAETAL03,LUNTETAL06}, to distinguish functional from
nonfunctional sequences \citep{KELLETAL03,GUIGETAL03,SIEPETAL07}, and to
detect sequences likely to be responsible for phenotypic differences
between species \citep{POLLETAL06,PRABETAL08}.  


Most evolutionary analyses of noncoding elements so far have made use of
sequence conservation between genomes that diverged millions of years ago.
Many confounding factors limit the utility of these approaches, including
turnover of regulatory elements \citep{DERMCLAR02,MOSEETAL06,SCHMETAL10},
challenges in orthology identification, and alignment error.  In principle,
data describing genetic variation could help to address these limitations,
because it reflects evolutionary processes on much shorter timescales,
during which turnover should be much less prevalent.  Orthology
identification and alignment are also much more straightforward on these
time scales.  It is well known that patterns of polymorphism within a
species and divergence between species can be used to tease apart the
effects of positive selection, negative selection, and neutral drift for a
given collection of functional elements
\citep{MCDOKREI91,SAWYHART92,BUSTETAL05}.  In practice, however, it is
technically challenging to extract useful information about noncoding
elements from patterns of polymorphism and divergence for various reasons.
Many noncoding elements of interest, such as transcription factor binding
sites, are quite short (typically $<$10 bp) and polymorphisms tend to be
sparse, so that most elements contain no informative sites.
Furthermore, factors such as variation across loci in mutation rates and
time to most recent common ancestry, and the influence of demography on
patterns of polymorphism, make it difficult to interpret patterns of
polymorphism and divergence in noncoding elements, and prohibit
straightforward pooling of data from multiple elements across the genome.

Here we describe a new computational method, called \underline{I}nference
of \underline{N}atural \underline{S}election from \underline{I}nterspersed
\underline{G}enomically co\underline{H}erent elemen\underline{T}s (\ins),
that is designed to address these challenges.  \ins uses the general
strategy of contrasting patterns of polymorphism and divergence in a
collection of elements of interest with those in flanking neutral regions,
thereby mitigating biases from demography, variation in mutation rates, and
differences in genealogical backgrounds.  In this way, it resembles
McDonald-Kreitman-based methods for identifying departures from neutrality
\citep{MCDOKREI91,ANDO05,SAWYHART92,SMITEYRE02}.  Unlike these methods,
however, \ins is based on a generative probabilistic model, accommodates
weak negative selection \citep{CHAREYRE08}, and allows diffuse information
from many short elements across the genome to be pooled efficiently, in a
manner that avoids statistical pitfalls arising from pooling counts of site
classes \citep{SMITEYRE02,STOLEYRE11}.  Our modeling approach accommodates
variable mutation rates and times to most recent common ancestry
along the genome and fully integrates phylogenetic information from
multiple outgroup species with genome-wide population genetic data.  In
other recent work, we have applied \ins in a large-scale analysis of
transcription factor binding sites in the human genome, using chromatin
immunoprecipitation and sequencing (ChIP-seq) data for 78 human
transcription factors (TFs) from the ENCODE project \citep{BERNETAL12} and
54 unrelated complete human genome sequences from Complete Genomics
(\url{http://www.completegenomics.com/public-data/69-Genomes/})
\citep{DRMAETAL10}.  Our focus 
in this article is to detail the probabilistic model and inference strategy
underlying the method, and examine its performance on simulated data under
a range of scenarios.  In addition, we provide an analysis of several
additional classes of noncoding elements, a more detailed analysis of GATA2
binding sites, and a detailed analysis of individual positions within
primary micro-RNA transcripts.

\section*{Methods}

\subsection*{General Approach}

Our method is designed to measure the influence of natural selection on a
collection of genomic elements scattered across the genome (Fig.\
\ref{fig:model-schematic}).  The collection of interest could be defined in
various ways; for example, it could consist of all binding sites of a
particular transcription factor, all noncoding RNAs of a particular type,
or a subset of interest, such as binding sites near genes
of a particular functional category (see Discussion).  We assume the
individual elements are quite short---typically only a few nucleotides in
length, and not longer than a few kilobases.  The key modeling challenge
is to integrate sparse information from many such elements in a manner that
accounts for variation along the genome in properties such as mutation
rates and local genealogical structure.  (Note that, even with constant
mutation rates, regions will differ in their patterns of polymorphisms due
to differences in times to most recent common ancestry and other
properties.)  Rather than attempting to fully describe the relationships
among selection, polymorphism, and divergence---which is complex and
demography-dependent---our model works by contrasting patterns of
polymorphism and divergence in the elements of interest with those in
putatively neutral sites nearby.

We assume genome-wide polymorphism data is available for a particular
target population, in a form that allows polymorphic sites to be reliably
distinguished from invariant sites, and that provides reasonably accurate
information about allele frequencies.  At present, this is most easily
achieved using high-coverage individual genome sequences, although our
methods could also be adapted make use of statistically inferred genotype
frequencies based on low-coverage sequence data \citep{YIETAL10}.  
We further assume genome-wide data is
available for one or more outgroup species, typically in the form of
reference genome assemblies.  While the method can be used with a single
outgroup genome, it is highly desirable to make use of two or more
outgroups that diverged from one another prior to the divergence of either
from the target population.  In addition, the outgroup sequences generally
should be 
as closely related to the target population as possible.  This will ensure
the highest quality information about ancestral alleles for the target
population.  

We use a categorical model for the distribution of fitness effects (DFE).
Specifically, we 
assume each nucleotide site evolves according to one of four possible
selective modes: neutral drift (neut), strong negative selection (SN), weak
negative selection (WN), or strong positive selection (SP).  (Sites under
weak positive selection are assumed to be rare and are absorbed in the
neutral category.) This
coarse-grained approach is motivated by observations indicating that the
data contain only limited information about the full distribution of
fitness effects \citep{BOYKETAL08,WILSETAL11}.
These categories are chosen for having qualitatively distinct effects on
patterns of polymorphism and divergence \citep[see][]{BIEREYRE04}.  In
particular, our model makes use of the fact that strong selection (negative
or positive) generally causes mutations to be eliminated or reach fixation
rapidly,
while weak negative selection allows polymorphisms to persist for longer
periods of time, but tends to hold derived alleles at low
frequencies. Therefore, we assume that at nucleotide sites under selection,
(1) only SP sites make nonnegligible contributions to divergence, (2) only
WN sites make nonnegligible contributions to polymorphism, and (3) any
polymorphisms must have low derived allele frequencies.  Together, these
assumptions allow the fraction of sites under selection to be estimated.
As it turns out, they are not sufficient to fully disentangle the
contributions of all four selective modes, but they do allow us to obtain
indirect information about the contributions of positive and weak negative
selection at selected sites (see below).

In addition, our model reduces the site frequency spectrum (SFS) to three
classes: every site is considered monomorphic (M), polymorphic with a low-frequency
derived allele (L), or polymorphic with a high-frequency derived allele
(H), where the distinction between L and H sites depends on a designated
low-frequency threshold $f$ (typically $f=0.15$).  Information about
selection comes from the relative frequencies of these labels in the elements
of interest relative to the flanking neutral sites, together with patterns of
divergence with respect to the outgroup genomes.  A minor complication is
that in some cases, the derived allele class depends on the ancestral
allele, which is not known.  We address this problem by treating the
ancestral allele as a hidden (latent) random variable and integrating over possible
values as needed.  The use of a low-dimensional projection of the SFS is
intended to buffer our method from the effects of recent
demographic changes in the target population.  In the simulation analyses
reported below, we examine the extent to which our inferences are robust to
demography.  We also examine their dependence on the 
threshold~$f$.


\subsection*{Probabilistic Model}

Our model assumes that the genomic regions under study are partitioned into a
collection of blocks, $B$.  The nucleotide sites within
each block $b \in B$ are further partitioned into sites within the elements
of interest, $E_b$, and the associated neutral flanking sites, $F_b$
(cumulatively $E$ and $F$, respectively).  Each block is assigned a
population-scaled mutation rate ($\theta_b$), a neutral divergence scale
factor ($\lambda_b$), and an outgroup divergence scale factor
($\lambda^O_b$).  In addition, the model has four global parameters: the
fraction of sites under selection in elements ($\rho$), the relative
divergence ($\eta$) and polymorphism ($\gamma$) rates at selected sites,
and $\vect\beta$, a multivariate parameter summarizing the neutral site
frequency spectrum (see Table \ref{tab:parameters}).  The full set of
parameters is denoted $\vect\zeta$.  

Each site $i$ is associated with a set of aligned bases from outgroup
genomes ($O_i$) and the polymorphism data for the target population
($X_i$).  $X_i$ is further summarized as $X_i=(X_i^{\text{maj}},
X_i^{\text{min}}, Y_i)$, where $X_i^{\text{maj}}$ and $X_i^{\text{min}}$
are the observed major and minor alleles, and $Y_i \in $\{M, L, H\} is the
minor frequency class ($X_i^{\text{min}}=\emptyset$ when $Y_i= \text{M}$).
The entire data set is denoted by $(\data, \outgroup)$.  $Y_i$ is defined
by the observed minor allele frequency $m_i$ and the specified
low-frequency threshold, $f<\frac 12$; in particular, $Y_i=\text{M}$ when
$m_i=0$, $Y_i=\text{L}$ when $0<m_i<f$, and $Y_i=\text{H}$ when $m_i\geq
f$.  Sites with three or more alleles are discarded in pre-processing.
Each site is associated with three hidden variables: a selection class
($S_i\in\{\sel,\neut\}$), a ``deep'' ancestral allele at the most recent
common ancestor of the target population and closest outgroup ($Z_i$), and
a population ancestral 
allele ($A_i$) (Table \ref{tab:variables}).  In addition, when $Y_i =
\text{L}$, the model has to consider uncertainty in the derived allele
frequency class, which could be L or H, depending on the identity of the
ancestral allele.

We assume independence of blocks, conditional independence of the
nucleotide sites within each block given the model parameters, conditional
independence of the variables describing the target population ($A_i$,
$S_i$, and $X_i$) and the outgroups ($O_i$) given the deep ancestral allele
$Z_i$, and independence of the $S_i$ values given the parameter
$\rho$ (as shown graphically in Fig.\ \ref{fig:graphical-model}).  The same
graphical model applies to all sites, except that the
selection class is fixed to ``neut'' for the flanking sites.  Thus, a
likelihood function for the model, conditional on the outgroup data, can be
written as follows:
\begin{small}
\begin{align}\label{eqn:likelihood}
\ld(\vect \zeta\ ; \data,\outgroup) ~~&\equiv~~ P(\vect X\  |\  \vect O, \vect \zeta)
~~= \notag \\
\prod_{b\in B}~ 
&\left[\prod_{i\in F_b}
\sum_z \sum_a 
P(X_{i}, Z_{i}=z, A_{i}=a \ |\ S_{i}=\neut, 
O_{i}, \vect\zeta) \right] \notag \\
\times &\left[\prod_{i\in E_b}\sum_{s\in \{\neut,\sel\}}P(S_{i}=s\ |\ \vect\zeta)\
\sum_z \sum_a  P(X_{i}, Z_{i}=z, A_{i}=a \ |\  S_{i}=s, O_{i}, \vect\zeta)\right]~.
\end{align} \end{small}
Furthermore,
each term of the form $P(X_{i}, Z_{i}, A_{i} \ |\ S_{i}, O_{i},\vect\zeta)$ 
can be factorized as follows:
\begin{align}\label{eqn:sitewise-likeihood}
P(X_{i}, &Z_{i}, A_{i} \ |\  S_{i},
O_{i}, \vect\zeta) ~=~
P(Z_{i}\ |\ O_{i},\ \lambda^O_b)~
P(A_{i}\ |\ S_i, Z_{i},\ \vect\zeta)~
P(X_{i}\ |\ S_i, A_i, Z_i,\ \vect\zeta) 
\end{align}

This likelihood function is composed of four conditional probability
distributions, corresponding to the variables $S_i$, $Z_i$, $A_i$, and
$X_i$.  The  
distribution for $S_i$ is needed only for element sites and is
given by a two-component mixture 
model with coefficient $\rho$:
%
%
\begin{small}
\begin{equation}\label{eqn:mixture}
P(S_i=s\ |\ \vect\zeta) = \begin{cases}
\rho & s = \text{\sel}\\
1-\rho & s = \text{\neut}\\
\end{cases}~.
\end{equation}
\end{small}

The conditional distribution for $Z_i$ given the outgroup data, $P(Z_{i}\
|\ O_{i},\ \lambda^O_b)$, is based on a standard statistical phylogenetic
model and is computed using existing software.  Notice that our model
assumes that the phylogenetic model 
for the outgroups is independent of the selection class, $S_i$.  This
assumption is not strictly warranted (sites under selection are likely to
evolve at different rates in the outgroups), but it dramatically simplifies
the inference procedure by allowing us to pre-estimate the outgroup scale
factors ($\lambda^O_b$) and the sitewise distributions for $Z_i$ (see {\bf
  Parameter Inference}). It also allows us to avoid specifying a model for
the poorly understood process of turnover of functional elements.  In
practice, this
simplifying assumption is of little consequence, 
because it only affects the prior distribution for $Z_i$, which is fairly
insensitive to evolutionary rates in outgroup lineages as long as
the branches of the phylogeny are not too long.



The third conditional distribution, $P(A_{i}\ |\ S_i, Z_{i},\ \vect\zeta)$,
describes the process of sequence divergence on the lineage leading to the
target population.  Given a global neutral branch length $t$ for this
lineage (in substitutions per site), we assume a nucleotide substitution
rate of $\lambda_bt$ for neutral sites and $\eta\lambda_bt$ for sites under
selection.  Note that $\eta$ can be driven downward by negative selection
or upward by positive selection so it may be greater or less than one,
depending on the DFE.  In principle, any DNA substitution model could be
used to define $P(A_{i}\ |\ S_i, Z_{i},\ \vect\zeta)$.  However, because we
are primarily interested in cases in which $t$ is quite small (e.g., $t
\approx 0.005$ for the case of humans and chimpanzees), we assume a
Jukes-Cantor or Poisson substitution model and approximate the divergence
probabilities as:
\begin{small}
\begin{equation}\label{eqn:div}
P(A_i=a\ |\ S_i=s, Z_i=z, \vect \zeta) ~=~
\begin{cases}
\frac13 \lambda_b t & s=\text{\neut},\ a \ne z \\
1- \lambda_b t & s=\text{\neut},\ a = z \\
\frac13 \eta\lambda_b t & s=\text{\sel},\ a \ne z \\
1- \eta\lambda_b t & s=\text{\sel},\ a = z \\
\end{cases}
\end{equation}
\end{small}

Finally, the fourth conditional distribution, $P(X_{i}\ |\ S_i, A_i, Z_i,\
\vect\zeta)$, describes the patterns of polymorphism in the target
population given the ancestral alleles and selection class.  In deriving
this expression, we first assume an infinite sites model for the time since
the population-level MRCA (which is expected to be much shorter than the
time since the deep ancestral allele), implying that
$A_i\in\{X_i^{\text{maj}},X_i^{\text{min}}\}$.  The neutral
population-scaled mutation rate is given by $\theta_b=4N_b\mu_b$, where
$N_b$ is a hypothesized block-specific effective population
size.  Because $\theta_b$ is estimated freely, without any constraints on
$\mu_b$ or $N_b$, the model can accommodate sources of variation in
nucleotide diversity other than variable mutation rates, such as selection
from linked sites (i.e., background selection or hitchhiking).  Sites under
selection ($S_i = \text{sel}$) are assumed to have a population-scaled
mutation rate of $\gamma\theta_b$ (we expect, but do not require,
$\gamma<1$). Given a population-scaled mutation rate of $\theta_b$, the
probability of observing a polymorphic sites in a sample of 
size $n$ is given by $\theta_b a_n$, where $a_n = \sum_{k=1}^{n-1}1/k$
\citep{WATT75}. In the absence of missing data, $a_n$ is a constant of no
consequence in the inference procedure, but it can be used to accommodate
missing data if desired (see Discussion and Supplementary Methods).

Under neutrality, the polymorphism and divergence components of the model
are assumed to be independent; i.e., $X_i$ and $Z_i$ are conditionally
independent given $A_i$ and $S_i=\text{neut}$.  In this case, the derived
allele in a polymorphic site is assumed to be chosen uniformly at random
from the three bases not equal to $A_i$, and the derived allele frequency
class is assumed to be chosen at random from the three intervals $(0,f)$,
$[f,1-f]$, or $(1-f,1)$ with probabilities $\beta_1,\beta_2$, and
$\beta_3$, respectively.  The distinction between the two high-frequency
classes, $[f,1-f]$ and $(1-f,1)$, is required because they correspond to
different minor allele frequency classes ($Y_i=\text{H}$ and
$Y_i=\text{L}$, respectively), which, in general, will have different
probabilities.  Thus, in the case of $S_i=\text{neut}$, the conditional
probability for $X_i$ is given by:

%
%
\begin{small}
\begin{align}\label{eqn:poly-neut}
P\left(X_i=(x^{\text{maj}},x^{\text{min}},y) \ |\ S_i=\neut, A_i=a, Z_i,\
  \vect\zeta\right) &= \notag \\
P\left(X_i=(x^{\text{maj}},x^{\text{min}},y) \ |\ S_i=\neut, A_i=a,\ \vect\zeta\right) &=
\begin{cases}
1- \theta_b a_{n} & y = \text {M},\ a=x^{\text{maj}} \\
\frac 13 \beta_1 \theta_b a_n & y = \text {L}\ ,\ a=x^{\text{maj}} \\
\frac 13 \beta_3 \theta_b a_n & y = \text {L}\ ,\ a=x^{\text{min}} \\
\frac 13 \beta_2 \theta_b a_n & y = \text {H}\ ,\ a\in \{x^{\text{maj}},x^{\text{min}}\} \\
0 & \text{otherwise}
\end{cases}
\end{align}
\end{small}

The model for polymorphism at selected sites is similar, but incorporates our two main assumptions
regarding sites under selection, namely that polymorphisms are restricted to WN sites, implying that
they do not occur in sites that have experienced divergence (hence the conditional dependence in $Z_i$),
and that WN polymorphisms have low derived allele frequencies, implying that $Y_i=\text{L}$ and $X^{\text{maj}}_i=A_i$:
%
\begin{small}
\begin{align}\label{eqn:poly-sel}
P\left(X_i=(x^{\text{maj}},x^{\text{min}},y)  |\ S_i=\sel, A_i=a, Z_i=z,\ \vect\zeta\right) &=
\begin{cases}
1- \gamma\theta_b a_n & y = \text {M},\ z = a = x^{\text{maj}} \\
1 &  y = \text {M},\ z\neq a = x^{\text{maj}} \\
\frac 13 \gamma \theta_b a_n &  y = \text {L},\ z = a = x^{\text{maj}} \\
0 & \text{otherwise}
\end{cases}
\end{align}
\end{small}

Finally, the models for polymorphism and divergence (Equations
\ref{eqn:div}--\ref{eqn:poly-sel}) can be
combined into a single conditional distribution
table, $P(X_{i}\ |\ Z_{i}, S_{i}, \vect \zeta)$ (Table \ref{tab:cond-dist}),
by integrating over the cases of $A_i\in\{X_i^{\text{maj}},X_i^{\text{min}}\}$.


\subsection*{Parameter Inference}\label{subsec:model-infer}

The main objective of the inference procedure is to produce maximum
likelihood estimates (MLEs) of the selection parameters, $\rho$, $\eta$,
and $\gamma$, but in order to do so, the neutral parameters
$\vect{\zeta_{\text{\neut}}}=\left(\vect{\lambda^O},\vect{\lambda},
  \vect{\theta},\vect{\beta}\right)$ must also be estimated.  In principle, an
expectation-maximization (EM) algorithm could be used to jointly estimate all
model parameters.  However, this approach is impractical for genome-wide
applications involving millions of nucleotide sites.  Instead, we take
advantage of the ``loose coupling'' between the phylogenetic outgroup model and the remaining portions of the model, and between the portions of the
model concerned with the elements and the flanking sites, to decompose
the inference procedure into separate stages, each of which can be
performed fairly simply and efficiently. 

First, observe that the likelihood function can be viewed as a product of
a function
of the flanking sites and a function of the element sites.  The first
function does not depend on the selection parameters. 
Moreover, if the flanking sites significantly outnumber the neutral sites
within the elements, as we expect, then the neutral parameters can be
estimated to a good approximation by maximizing this 
function only.  The selection parameters can then be estimated by
conditionally maximizing the second function.
More precisely, we represent the likelihood function as:
\begin{equation}
\ld(\vect \zeta\ ; \data,\outgroup) = \ld_F(\vect{\beta},\vect{\lambda^O},\vect{\lambda},\vect{\theta}\
;\ \data_F,\outgroup_F) \times \ld_E(\rho, \eta, \gamma\ ;\ \data_E,\outgroup_E,
\vect{\beta},\vect{\lambda^O},\vect{\lambda},\vect{\theta}), 
\end{equation}
where
\begin{small}
\begin{align}
\ld_F(\vect{\beta},\vect{\lambda^O},\vect{\lambda},\vect{\theta}\
;\ \data_F,\outgroup_F) &~=~   \prod_{b\in B} ~\prod_{i\in F_b} P(X_{i} \ |\
S_{i}=\text{\neut},\ O_{i},\ \lambda^O_b,\lambda_b,\theta_b,
\vect{\beta})~, \label{eqn:likelihood-neut} \\
\ld_E(\rho, \eta, \gamma\ ;\ \data_E,\outgroup_E,
\vect{\beta},\vect{\lambda^O},\vect{\lambda},\vect{\theta}) &~=~ \prod_{b\in
  B} ~\prod_{i \in E_b} \sum_{s \in {\text{neut, sel}}}P(S_i=s|\rho)\,
P(X_{i} \ |\ S_{i}=s,\ O_{i},\ \eta, \gamma,
\lambda^O_b,\lambda_b,\theta_b,  \vect{\beta}), \label{eqn:likelihood-sel}
\end{align} \end{small}

\vspace{-4ex}
\noindent and we estimate the neutral parameters by maximizing Equation
\ref{eqn:likelihood-neut}, then estimate the selection parameters by
conditionally maximizing Equation \ref{eqn:likelihood-sel}.


Furthermore, the
phylogenetic and population genetic parameters in Equation \ref{eqn:likelihood-neut}
can be estimated separately by making some additional minor simplifying
assumptions. Briefly, the divergence scale factors 
$\lambda_b$ and $\lambda^O_b$ are first estimated by fitting a
pre-estimated neutral phylogenetic model to putative neutral sites in each
genomic block using standard phylogenetic fitting procedures
\citep{HUBIETAL11} (see Supplementary Methods).  The fitted
phylogenetic model is then used to compute the prior distribution for
ancestral alleles,
$P(Z_i\ |\ O_i,\lambda^O_b)$, at all sites in the block (including $E_b$),
conditioning on the outgroup sequences only.  Next, maximum likelihood
estimates of the block-specific polymorphism rate parameters,
$\hat{\theta}_b$, and the global parameter $\beta_2$ are obtained using
simple closed-form expressions (see Supplementary Methods). Due to
uncertainty about the ancestral allele, $\beta_1$ and $\beta_3$ do not have
closed-form estimators, and are estimated by a simple EM algorithm.

Finally, the selection parameters are estimated
conditional on the neutral parameters by maximizing Equation
\ref{eqn:likelihood-sel} by EM.
To derive this algorithm, let us first imagine that all variables are
observed, and denote by
$c_Q(\mathcal{X})$ the number of sites in a set $Q$ that have a
configuration $\mathcal{X}$.  Using this notation, the complete-data 
log-likelihood function for selected sites can be expressed as:
\begin{small}
\begin{align}
\ln [\ \ld_E(\rho, \eta, \gamma& ;\ \data_E,\outgroup_E,
\vect{\hat\beta},\vect{\hat\lambda^O},\vect{\hat\lambda},\vect{\hat\theta})\ ]
 = \notag \\
\notag &c_E(S_i=\text{\sel})\ln(\rho) ~+~c_E(S_i=\text{\neut})\ln(1-\rho) ~~+\\
\notag &c_E(S_i=\text{\sel},Z_i\neq X_i^{\text{maj}})\ln(\eta) ~+~ \sum_{b\in B}c_{E_b}(S_i=\text{\sel},Z_i= X_i^{\text{maj}})\ln(1-\eta\lambda_bt) ~~+ \\
&c_E(S_i=\text{\sel},Y_i=\text{L})\ln(\gamma) ~+~ \sum_{b\in B}c_{E_b}(S_i=\text{\sel},Y_i=\text{M},Z_i= X_i^{\text{maj}})\ln(1-\gamma\theta_ba_n) ~~+~~C~,~
\label{eqn:sel-likelihood}
\end{align}
\end{small}
%
%

\vspace{-4ex}
\noindent where $C$ is a constant term that does not depend on $\rho,
\eta$, or $\gamma$.  In practice, 
the counts in Equation \ref{eqn:sel-likelihood} depend on
the hidden variables $Z_i$ and $S_i$, so this function must be iteratively
optimized by EM.  Briefly, the E step in each iteration uses the current
estimates of the model parameters to obtain expectations of these
counts by computing the sitewise posterior probabilities of the variable
configurations $(S_i=\neut)$, $(S_i=\sel, Z_{i}= X^{\text{maj}}_{i})$, and
$(S_i=\sel, Z_{i}\neq X^{\text{maj}}_{i})$.  These calculations make use of
the conditional probabilities in Table \ref{tab:cond-dist} and the
pre-computed prior distributions, $\{P(Z_i\ |\ O_i,\hat\lambda^O_b)\}_{i\in
  E}$. After obtaining the expected
counts, the M step updates the selection parameters $\rho,\eta$, and
$\gamma$ to values that maximize the expected log-likelihood function. The
update for $\rho$ is achieved using a close-form expression, while the
updates for $\eta$ and $\gamma$ require numerical optimization of a concave
function (see Supplementary Methods).

\subsection*{Extracting Information about the Modes of Selection}
\label{subsec:interpret}

While the model does not permit direct estimation of the fractions of sites
under weak negative (WN), strong negative (SN), or strong positive (SP)
selection, it can be used to obtain indirect measures of the impact of WN
and SP selection.  In particular, a useful measure of SP selection is
$\PD$, the number of divergence events driven by positive selection
(sometimes called ``adaptive substitutions'') on the branch to the target
population.  A similar measure pertaining to WN selection is $\WP$, the
number of polymorphic sites subject to selection. Expected values for $\PD$
and $\WP$ can be obtained by summing over site-wise posterior probabilities
associated with the variable configurations $(Y_i=\text{M},Z_i\neq A_i,
S_i=\sel)$ and $(Y_i=\text{L}, S_i=\sel)$, respectively (see Supplementary
Methods). These calculations make use of our assumptions that, among
selected sites, divergence events occur only due to SP selection, and
polymorphisms occur only due to WN selection and are restricted to
frequency class `L'.  In our analysis, we normalize $\Ea$ and $\Ew$ by
dividing them by the number of nucleotide sites considered (in kilobases),
to allow comparisons between sets of different sizes.  By dividing $\Ea$ by
the total (expected) number of 
divergences, one can alternatively obtain an estimate of the fraction of
substitutions driven by positive selection, a quantity known as $\alpha$
\citep[e.g.,][]{SMITEYRE02,ANDO05} (see Supplementary Methods).

\subsection*{Confidence Intervals and Likelihood Ratio Tests}

Standard errors for the estimated selection parameters were estimated using
the curvature method \citep{LEHMCASE98}, based on an approximate Fisher
information matrix derived from the $3 \times 3$ matrix of second
derivatives for the log-likelihood function for $\rho$, $\eta$, and
$\gamma$ (Equation \ref{eqn:likelihood-sel}) at the joint MLE (see
Supplementary Methods).  In addition, we used likelihood ratio tests (LRTs)
to evaluate evidence for selection in general ($\rho>0$), positive
selection ($\eta>0$), and weak negative selection ($\gamma>0$).  The LRTs
were performed by fitting the model to the data twice, once with no
restrictions on the free parameters, and once with a parameter of interest
fixed at zero.  Twice the difference in log likelihoods was then treated as
a test statistic and compared to an appropriate asymptotic distribution.
The tests for $\eta>0$ and $\gamma>0$ involve nested models in which the
null hypothesis falls at a boundary of the alternative hypothesis.  The
associated test statistics therefore have asymptotic null distributions
equal to a 50:50 mixture of a $\chi^2$ distribution with one degree of
freedom (dof) and a point mass at zero \citep{CHER54,SELFLIAN87}.  The case
of $\rho$ is more complex, because a value of $\rho=0$ causes $\eta$ and
$\gamma$ to become irrelevant to the likelihood function, violating the
regularity conditions for the asymptotic mixture distribution.  Still, it is
reasonable to expect that the asymptotic distribution will be approximately
given by a 50:50 mixture of a $\chi^2$ distribution with 3 dof and a point
mass at zero.  These asymptotic distributions are of course not guaranteed
to hold 
for real data sets, and we use them only for approximate
assessments of statistical significance (see {\bf Results}).


\subsection*{Implementation and Software}

The \ins software consists of several modules.
The main module is a C program implementing the 
EM algorithm for inference of the selection parameters, as well as the
simpler EM algorithm for estimation of $\beta_1$ and $\beta_3$. This
program outputs maximum likelihood estimates of $\rho$, $\eta$, and
$\gamma$, the posterior expected values $\Ea$, $\Ew$, and $\alpha$, and
approximate standard errors for the reported values.
The phylogenetic model fitting stage is implemented separately using
procedures from RPHAST \citep{HUBIETAL11}, and additional scripts are used
for processing and filtering the polymorphism data.  Source code,
documentation and sample files are available for download from
\url{http://compgen.bscb.cornell.edu/INSIGHT/}.

\subsection*{Simple Site-count-based Estimates}

For comparison with our model-based estimates, we made use of simple
estimators for the fraction of sites under selection ($\rho$) and the
number of adaptive substitutions ($\PD$).  These estimators are based on
the numbers of 
polymorphisms in element and flanking sites, denoted $P_E$ and $F_E$,
respectively, and the numbers of divergence events in element and flanking
sites, denoted $D_E$ and $D_F$, respectively.  They include a
divergence-based estimator for $\rho$ introduced by \cite{KONDCROW93},
\begin{equation}
\hat\rho_{\text{Div}} ~=~~ 1 - \frac{D_{E}\ |F|}{|E|\ D_{F}},
\label{eqn:div-rho}
\end{equation}
a parallel estimator based on polymorphism rates,
\begin{equation}
\hat\rho_{\text{Poly}} ~=~~ 1 - \frac{P_{E}\ |F|}{|E|\ P_{F}},
\label{eqn:poly-rho}
\end{equation}
and an estimator for $\Ea$ based on the
McDonald-Kreitman \citeyearpar{MCDOKREI91} test, adapted from
\cite{SMITEYRE02}:
\begin{equation}
\hat\PD_{\text{-MK}} ~~=~~ 
D_E\ - \frac{P_{E}\ D_{F}}{P_{F}}~.
\label{eqn:mk-pd}
\end{equation}
In comparison with our model based estimates, the divergence-based
estimator $\hat\rho_{\text{Div}}$ ignores the effect of positive
selection, and the estimators $\hat\rho_{\text{Poly}}$
and $\hat\PD_{\text{-MK}}$ both implicitly assume no polymorphisms occur in
selected sites, thus ignoring the effects of WN selection. All
three estimators share the limitation of pooling counts across elements in
a manner that does not account for variable mutation rates across loci.

\subsection*{Simulations}

Simulated elements and flanking regions were generated with the forward
simulator SFS\_CODE \citep{HERN08}, assuming various mixtures of selective
modes for the elements.  We simulated data for human populations and
chimpanzee, orangutan, and rhesus macaques outgroups, using parameters
based on previous studies.  Each simulated block consisted of a 10 bp
element, reflecting a typical binding site, and 5,000 flanking neutral
sites on each side.  We assumed a constant recombination rate and a
randomly varying mutation rate, and each nucleotide position was assigned
to one of four selection classes: neutral evolution ($2N_es=0$), strong
negative selection ($2N_es=-100$), weak negative selection ($2N_es=-10$),
and positive selection ($2N_es=10$).
Our choices of population-scaled selection coefficients were approximately based
on several other recent studies
\cite[e.g.,][]{EYREETAL06,BOYKETAL08,WILSETAL11}.  
Selection at WN and SN sites was held
constant across the phylogeny, while for SP sites we assumed an interval of
positive selection followed by weak negative selection on the lineage
leading to the human population, to simulate selective sweeps rather than
recurrent positive selection (see Supplementary Methods for complete
details).  The 10 kb flanking sites were all assigned to the neutral class,
and the 10 bp of each simulated element were allocated among the four
classes by multinomial sampling.  In addition to assuming a range of
mixtures of selective modes, we considered scenarios with various numbers
of elements (ranging from 10,000--20,000).

The values of $\rho$, $\Ea$ and $\Ew$ estimated by \ins were compared with
``true'' values for each simulation.  The true value of $\rho$ was simply
the fraction of sites assumed to be under selection during data generation.
The true value of $\PD$ was taken to be the number of actual divergence
events that occurred in sites under positive selection.  The true value of
$\WP$ was taken to be the number of negatively selected sites that are
polymorphic.  In computing this quantity we allowed for both strong and
weak negative selection, because the distinction between them is somewhat
arbitrary in our model.  For $\rho$ and $\PD$ we also compared our
model-based estimates with the simple estimates based on counts of
polymorphic and divergent sites (Equations
\ref{eqn:div-rho}--\ref{eqn:mk-pd}).

\subsection*{Analysis of Human Noncoding Genomic Elements}

In our experiments on real data, we made use of the 69 individual human
genome sequences recently released by Complete Genomics
(\url{http://www.completegenomics.com/public-data/69-Genomes/}) \citep{DRMAETAL10}, using data
for 54 unrelated individuals.  While larger data sets are available
\citep{1KGCONS10}, this one was selected for its high coverage, which
reduces the effect of genotyping error and allows singleton variants to be
characterized with fairly high confidence.  For outgroup genomes, we used
the chimpanzee (panTro2), orangutan (ponAbe2), and rhesus Macaque (rheMac2)
reference genomes.  Various filters were applied to guarantee high quality
alignments and variant calls (see Supplementary Methods).  Putatively
neutral sites were identified by excluding exons of known protein-coding
and RNA genes plus 1kb of flanking sites on each side, and previously
predicted conserved noncoding elements plus flanking regions of 100 bp.
After these filters were applied, average of 3,881 sites per 10,000 bp
block remained. Genomic blocks with $<$100 putative neutral sites were
discarded.

We examined several classes of short interspersed noncoding elements in the
human genome, including (1) several collections of regulatory noncoding
RNAs from GENCODE V.13 \citep{HARRETAL06} (See Supplementary Methods), (2)
proximal promoters of known genes (defined as 100 bp upstream the
transcription start site), and (3) a collection of GATA2 transcription
factor binding sites.  The GATA2 binding sites were identified by a
pipeline developed another recent study \citep{ARBIETAL12}, based on
genome-wide chromatin immunoprecipitation and sequencing (ChIP-seq) data
from the ENCODE project \citep{BERNETAL12}.
To improve efficiency, we performed the phylogenetic model fitting stage of
our analysis on a fixed set of 10kb genomic windows (overlapping by 5kb),
in a preprocessing step.  We fitted a neutral model estimated from fourfold
degenerate sites to the neutral sites in each window by estimating two
scale factors, one for the branch to the human genome ($\lambda_b$) and one
for the other branches in the tree ($\lambda^O_b$; see \cite{POLLETAL09}
for details).  This analysis assumed a (((human, chimpanzee), orangutan),
rhesus macaque) tree topology.  After fitting the model, we also computed
conditional distributions for the ancestral allele $Z_i$ given the outgroup
sequences at each nucleotide position $i$.  We also estimated $\theta_b$
for each block.  The estimates of $\lambda_b$ and $\theta_b$, and the
distributions for $Z_i$, were recorded in a database and used in all
subsequent analyses.

\section*{Results}

\subsection*{Simulations}

We applied \ins to various collections of synthetic elements and compared
our model-based parameter estimates both with ``true'' values reflecting
the simulated evolutionary histories and with various simple estimators
based on counts of polymorphisms and divergences (see {\bf Methods}).  We
simulated collections roughly similar to our real data sets (Supplementary
Table\ \ref{tab:element-classes}), with 10,000--20,000 blocks consisting of
10 bp elements and 10 kbp of flanking neutral sequence.  We considered a
range of mixtures of neutral, weak negative (WN), strong negative (SN), and
strong positive (SP) selection (see {\bf Methods}).  Here we focus on four
representative data sets: (1) one with relatively few sites under selection
(10\%) and negative selection only (`Neg' in Fig.\ \ref{fig:sim-results}A);
(2) another with a moderate fraction of sites under selection (30\%),
including a substantial fraction (5\%) under positive selection (`Pos');
(3) another with a high fraction of sites under weak negative selection
(50\%) and no sites under positive selection (`Weak'); and, finally, (4) a
set with a substantial fraction of sites in each of the selective modes
(`Mix').  We found that our model-based estimates of $\rho$ and $\PD$ were
highly accurate across all mixtures of selective modes.  Our estimates of
$\WP$ were also reasonably accurate, but had slightly larger confidence
intervals.  The simple estimators also performed reasonably well in many
cases, but the divergence-based estimators for $\rho$ were strongly biased
by positive selection (e.g., $\hat\rho_{\text{Div}}=-0.52$ in Pos and
$\hat\rho_{\text{Div}}=-0.13$ in Mix).  The reason for this bias is that
these estimators implicitly attribute all divergence to neutral drift, an
assumption that is violated by non-negligible levels of positive selection.
Similarly, the polymorphism-based estimator for $\rho$ was biased downward
in the presence of weak negative selection (e.g.,
$\hat\rho_{\text{Poly}}=0.59$ and $\rho_{\text{True}}=0.8$ in Mix), because
this estimators implicitly assumes that selection completely eliminates
polymorphism, which is not true in this case.  For similar reasons, the
McDonald-Kreitman (MK)-based estimates of the number of adaptive
divergences ($\hat\PD_{\text{-MK}}$) was also biased in the presence of
weak negative selection \cite[see][]{CHAREYRE08}.

These synthetic data sets---generated by forward simulation, under fairly
realistic assumptions---also enabled us to directly evaluate the
assumptions underlying our model.  Consistent with our assumptions, no
mutations reached fixation in the 34,000 negatively selected sites (weak or
strong) in our synthetic data sets.  Thus, our simulations strongly support
the critical assumption enabling the posterior expected number of
divergences under selection ($\Ea$) to be interpreted as a measure of
positive selection.  On the other hand, selected polymorphisms were not
completely restricted to WN sites, as assumed; instead, 8\% of
polymorphisms under selection occurred in SN sites and 9\% in SP sites,
with the remaining 83\% at WN sites.  However, the distinction our model
makes between WN and SN sites is inevitably somewhat arbitrary, and some
residual polymorphism in SN sites should have little impact on our
inference procedure.  (Indeed, it may be best to think of the WN sites as
being operationally defined as those negatively selected sites in which
polymorphisms are possible.)  On the other hand, the presence of some
polymorphic SP sites could lead to over-estimation of $\WP$, because these
sites will tend to be assigned to the WN class.  However, our inference
procedure appeared to be robust to minor violations of this assumption,
with no significant over-estimation of $\WP$.  Importantly, only a small
fraction (4\%) of all selected polymorphisms exhibited derived allele
frequencies $>$15\%, and these were vastly outnumbered by neutral
high-frequency polymorphisms.  Thus, while the simulation study did not
fully support our modeling assumptions, only fairly minor violations were
observed and our inference procedure seemed to be robust to them.

In the above analysis, we assumed a low-frequency threshold of $f=15$\%,
similar to previous studies \citep{FAYETAL01,ZHANGLI05,CHAREYRE08}.  In
reality, of course, the upper bound for the derived allele frequency at
negatively selected sites depends on various factors, including the actual
distribution of selection coefficients and the demographic history of the
sample.  To test the robustness of our model to the choice of $f$, we
generated eleven collections of 10,000 elements with true fractions of
sites under selection ranging from 0 to 1 (in steps of 0.1), keeping the
proportion within selected sites in each collection constant at 45\% WN,
50\% SN, and 5\% PD.  We then applied \ins to each data set using values of
$f$ ranging from 1\% to 40\% (Fig.\ \ref{fig:sim-results}B).  We found that
very low thresholds ($f<7$\%) resulted in clear under-estimation of all
model parameters, due to the presence of selected polymorphisms with DAF
exceeding the threshold, while very high thresholds ($f>20$\%) led to high
variance and some downward bias in the estimates, due to sparse data for
high-frequency polymorphisms.  Importantly, however, no bias was observed
for thresholds in the range of 7--20\%, indicating robustness to the
specific choice of threshold and justifying our default choice of 15\%.

An important feature of our model is that it directly contrasts sequence
patterns in elements with those in nearby neutral sites, which should make
it insensitive to the particular demographic history of the target
population. To test robustness to demography, we simulated data sets for
each of 
eleven mixtures of selective modes described above using four different
demographic 
scenarios for the target population: one with constant population size
since divergence from chimpanzee, one with a moderate population expansion,
and two others with a severe population bottleneck followed by an
exponential expansion (Supplementary Table\ \ref{tab:demography}).
Inference was performed separately for each of these 4$\times$11 data sets.  The
estimated parameters were then compared with their true values and with the
simple count-based estimates (Fig.\ \ref{fig:sim-results}C).  The
divergence-based estimates, $\hat\rho_{\text{Div}}$, were quite poor due to
the effects of positive selection, as discussed above. The
polymorphism-based estimates, $\hat\rho_{\text{Poly}}$, also consistently
under-estimated the true values, by an average of 24\% in the first two
(more moderate) scenarios and an average of 42\% in the scenarios with
bottlenecks. A more severe underestimation in the second two scenarios was
also observed for the MK-based estimator of the number of adaptive
divergences, $\hat\PD_{\text{-MK}}$.  In both cases, these underestimates
may reflect an increased influence from weak negative selection in
populations that have undergone bottlenecks.  In contrast, our model-based
estimates of $\rho$ and $\PD$ showed no apparent bias in any of the
simulated demographic scenarios.  Estimates of the number of polymorphisms
under selection, $\WP$, showed somewhat greater variance, as observed in
our initial simulation study (Fig.\ \ref{fig:sim-results}A), but the error
in these estimates did not seem to be affected by demography.  Thus, our
method appears to be capable of disentangling the contributions of positive
and negative selection even in the presence of a complex demographic
history, without the need for explicit demographic inference.

\subsection*{Analysis of Human Noncoding Genomic Elements}


We next applied \ins to real human genomic data, using 54 unrelated
individual genomes from Complete Genomics to 
define human polymorphisms and the chimpanzee, orangutan, and macaque
genomes as outgroups (see {\bf Methods}).  First, we applied the method to
randomly selected ``neutral'' regions (arbitrary genomic regions excluding
genes, conserved noncoding elements, and their immediate flanks; see {\bf
  Methods}), to ensure that it adequately controls for false
positive inferences of selection in real data.  From the previously
identified putatively neutral regions, we sampled 500 mutually exclusive
collections of roughly 30,000 ``neutral elements,'' 10 bp long.  For each
collection, we estimated $\rho$ and the corresponding likelihood ratio test
(LRT) statistic for the null hypothesis of $\rho=0$.  The 500 estimated
values of $\rho$ were generally close to zero, with a median of 0.03
(Supplementary Fig.\ \ref{fig:neutral-rho-hist}) and almost no values
$>$0.1.  The distribution of LRT statistics was roughly similar to a
50:50 mixture of a point mass at zero and a $\chi^2$ distribution with
three degrees of freedom, as expected (see {\bf Methods}), but did
show a clear shift toward large values relative to this distribution
(Fig. \ref{fig:genomic-results}A).
This shift may reflect violations of our simplifying assumptions in real
genomic data (e.g., variation in mutation rates within blocks),
contributions from alignment errors, or the inclusion of some functional
sites within our ``neutral'' elements.  Nevertheless, we found that the use
of a more conservative (non-mixed) $\chi^2$ distribution with
three degrees of freedom adequately controlled for the excess in large LRT statistics
(Fig. \ref{fig:genomic-results}A).  In particular, four of our data sets
(0.8\%) had LRT statistics that exceeded the $p=0.01$ cutoff and 24 data
sets (4.8\%) had LRT statistics that exceeded the $p=0.05$ cutoff,
indicating a good fit at the tail of the distribution.  Thus, we use this
distribution for approximate calculations of nominal $p$-values in our
subsequent analyses.

Next we examined five classes of noncoding elements annotated by the
GENCODE project.  These included proximal promoter regions of known genes
(defined as 100 bp upstream of the transcription start site), three classes
of noncoding RNAs (micro-RNAs [miRNAs], small nucleolar RNAs [snoRNAs], and
large interspersed non-coding RNAs [lincRNAs]), and binding sites of the
GATA2 transcription factor, in which we recently found evidence of both
positive and negative selection \citep{ARBIETAL12}.  We applied \ins to a
high-confidence subset of annotated elements in each of these five classes
(Supplementary Table \ref{tab:element-classes}; see Supplementary Methods).
Our analysis considered various thresholds for distinguishing between low
and high frequency polymorphisms, but our estimates were fairly insensitive
to this threshold (Supplementary Fig. \ref{fig:genomic-thresholds}), so we focus below on
results for the default threshold of 15\%.

All five classes of elements were estimated to have significant fractions
of sites under selection ($\rho>0$; $p\leq 0.01$; Fig.\
\ref{fig:genomic-results}B). The snoRNAs showed the highest estimated value
($\rho=0.46\pm 0.11$), consistent with their essential role in
guiding chemical modifications of ribosomal and transfer RNAs \citep{MATEETAL07,PANGETAL06}. 
miRNAs and GATA2 binding sites also showed estimates of $\rho$
exceeding 0.3, approximately the average for annotated transcription factor
binding sites \citep{ARBIETAL12}.  By contrast, lincRNAs were inferred to
have a much smaller (but still significant) fraction of sites under
selection, consistent with previous observations indicating high levels of
conservation are generally limited to short segments within lincRNAs
\citep{GUTTETAL09,MARQPONT09,ULITETAL11}.
We found significant evidence of weak negative
selection ($\gamma>0$; $p\leq 0.01$) for lincRNAs, snoRNAs, and proximal
promoters, with snoRNAs showing particularly high rates of weakly selected
segregating polymorphisms ($\Ew=1.7\pm 0.6$ polymorphisms per kbp).
Interestingly, only GATA2 binding sites showed significant evidence of
positive selection ($\eta>0$; $p\leq 0.01$) (Fig.\
\ref{fig:genomic-results}B), indicating that negative selection is dominant
for most of these noncoding elements, but at least some classes of
transcription factor binding sites exhibit substantial evidence of recent
adaptation \citep{ARBIETAL12}.  Estimates of $\rho$ naturally depend
on the density of functional sites within each annotation class, and the
reduced estimates
for promoters and lincRNAs likely reflect a relatively low density of
functional nucleotides.


To shed additional light on the manner in which natural selection has
influenced these elements, we performed a more detailed analysis of two
classes of elements, GATA2 binding sites and miRNA primary transcripts.
First, we partitioned the nucleotides in the annotated GATA2 binding sites
into 11 classes, corresponding to the 11 positions in the GATA2 motif, and
applied \ins separately to each class (Fig.\ \ref{fig:genomic-results}C).
This analysis indicated that the effects of natural selection are
concentrated in the seven-nucleotide ``core'' region of the motif (Fig.\
\ref{fig:genomic-results}C); all seven of these positions, and only one
other position, were found to have significant estimates of $\rho$ ($p\leq
0.01$).  Furthermore, it indicated that the signature of positive selection
comes primarily from the 7th and 8th positions, which together contribute a
total posterior expected number of adaptive divergences of 123$\pm$43
across roughly 30,000 binding sites (2.4$\pm$0.7 per kbp).  Interestingly,
these positions (particularly the 8th) are know to play a role in
modulating binding specificity of GATA2 \citep{KOENGE93,MERIORKI93}.  They
are also critical in determining the relative binding affinities of GATA1,
GATA2, and GATA3, which regulate overlapping sets of genes and are known to
serve as ``switches'' between alternative modes of gene expression
\citep{BRESETAL10,DOREETAL12}.  Notably, the 4th and 5th positions in the
core GATA motif showed significant evidence of weak negative selection
($\gamma>0$; $p\leq 0.01$), despite no significant signature of weak
negative selection in the global analysis of the binding site (above).

In our second detailed analysis, we partitioned the nucleotides in the
annotated miRNAs into several structural classes based on predictions of
hairpin secondary structures (see Supplementary Methods), and we applied
\ins separately to each class. We first partitioned the primary miRNA into
loop and stem regions (Fig.\ \ref{fig:genomic-results}D, inset)
, distinguishing between
paired and unpaired bases within the stem.  Among the three partitions,
only paired bases in the stem were estimated to have a significant fraction
of sites under selection ($\rho=0.48\pm 0.06$; $p< 10^{-5}$; Fig.\
\ref{fig:genomic-results}D), consistent with their key role in stabilizing
the hairpin structure.  The estimate for the unpaired stem positions was
particularly low ($\rho=0.15\pm0.12$; $p>0.05$).  These results are
consistent with previous comparative genomic studies in Drosophila
\citep{CLARETAL07,STARETAL07}.  The estimate for the loop region was
surprisingly high ($\rho=0.36\pm 0.22$ $p>0.05$), given that this region has no
known sequence-specific role in miRNA biogenesis, but data for the loop
was somewhat sparse, leading to high variance in the estimate.

Finally, 
we further partitioned the stem into four sub-regions---loop-proximal, lower stem, star, mature---reflecting the cleavage activity
of Drosha and Dicer, the two RNase III cleavage enzymes of primary
importance in miRNA biogenesis.
Estimates of $\rho$ for paired bases in these five regions
(Fig. \ref{fig:genomic-results}D) were generally concordant with previous
comparative analyses and with what is currently known about miRNA
biogenesis and target gene regulation
\citep{LAIETAL03,CLARETAL07,STARETAL07b}.  In particular, the highest
estimate of $\rho$ (0.66$\pm$0.15) corresponds to the 21--22nt mature
(miRNA) region, which has a dual role in structure preservation for
efficient recognition and processing by Drosha and Dicer, and in direct
post-transcriptional regulation of mRNA transcripts.  The lower-stem and
loop-proximal regions had lower 
estimates of $\rho$, probably because they do not serve any direct
regulatory role, but are important in preserving the hairpin structure.
The star region had an intermediate estimate of $\rho$, perhaps because
a fraction of star sequences are loaded into AGO complexes 
and carry out functional roles, even though most are degraded
\citep{OKAMETAL08,CLARETAL07}.
The estimates of $\rho$ obtained using \ins are generally similar to
comparative genomic estimates based on the phyloP program
\citep{POLLETAL09} (Supplementary Fig.\ \ref{fig:mirna-phylop}), but differ
from them in some respects.  For example, \ins finds somewhat weaker
evidence for selection in the star relative to the mature region of the
miRNA than does phyloP.  This difference could reflect a shift toward
weak negative selection in the star region, which is not apparent on
comparative genomic time scales 
because selection is sufficiently strong to prohibit long-term
fixation of derived alleles.

\section*{Discussion}


Methods based on patterns of divergence between species have become widely
used for identifying and characterizing noncoding functional elements
\citep{MARGETAL03,SIEPETAL05,COOPETAL05,POLLETAL09}, but these methods are
limited by their consideration of relatively long evolutionary time scales
and their sensitivity to alignment errors and other technical artifacts.
The goal of \ins is to shed new light on recent evolutionary patterns by
taking advantage of newly available population genomic data, together with
comparative genomic data for closely related species.  Any inference method
focused on recent evolutionary time must confront the problem that data
describing variation within populations and divergence on short time scales
is necessarily sparse.  \ins addresses this problem by considering
relatively large collections of elements of the same type, and directly
contrasting them with flanking neutral regions.  In this way, it accommodates
differences in mutation rates and genealogical backgrounds across the
genome, and mitigates biases from complex demographic histories.

\ins bears some similarities to McDonald-Kreitman (MK)-based methods
\citep{MCDOKREI91,SMITEYRE02,BIEREYRE04,ANDO05}, Poisson Random Field
(PRF)-based methods \citep{SAWYHART92,BUSTETAL02,BUSTETAL05,WILLETAL05},
and related methods for characterizing the distribution of fitness effects
\citep{EYREETAL06,BOYKETAL08,EYREKEIG09}, but it differs from previous
methods in several important respects.  Unlike MK-based methods, \ins is
based on a full generative probabilistic model, pools information from many
loci in a statistically rigorous manner, and explicitly models weak
negative selection.  Unlike PRF-based methods, it direct 
contrasts 
patterns of polymorphism and divergence in elements of interest with
flanking sites, rather than attempting to model the complex dependency of
absolute allele frequencies on selection coefficients.
\ins
additionally allows for straightforward likelihood ratio tests of various
hypotheses of interest, and it allows parameter variances to be
approximately characterized using standard methods.  For these reasons, we
expect it to be a valuable complement both to existing methods for
analyzing noncoding regions based on long-term evolutionary conservation,
and to methods for analyzing protein-coding sequences based on patterns of
polymorphism and divergence.

Our relatively simple probabilistic model is designed to
exploit newly available genome-scale data sets describing both candidate
functional elements \citep{BERNETAL12,GERSETAL10,ROYETAL10} and variation
within populations \citep{1KGCONS10,MACKETAL12}.  However, a naive approach
to parameter estimation would still be prohibitively CPU-intensive with
genome-wide data.  We achieve major gains in efficiency by decomposing the
inference procedure into three separate steps, concerned with the
estimation of the phylogenetic, neutral, and selection parameters,
respectively.  This decomposition relies on the simplifying assumption that
neutral sites within the elements of interest contain negligible
information about the neutral parameters of the model, because they are
vastly outnumbered by the flanking neutral sites---a property that can
typically be guaranteeed by construction.
It also depends on the use of a single phylogenetic model per locus in
estimating the prior distribution of the ancestral allele at all sites,
which should be adequate as long as relatively close outgroups are used.
Notably, the first two of these steps can be performed in preprocessing and
reused in the analysis of any set of loci that use the same flanking
regions.  Furthermore, the neutral flanks can be designed to maximize the
potential for reuse, as in this work, by defining a set of constant
genomic blocks, and associating each element with the neutral sites
of the nearest block.  This strategy allows the neutral and phylogenetic
parameters to be pre-estimated for each block and reused in any number of
subsequent analyses.  Importantly, these steps dominate the running time of
the inference algorithm
(particularly the phylogenetic estimation step).  The final stage, in which
the parameters $\rho$, $\eta$, and $\gamma$ are estimated, is independent
of the number of genomes considered and typically takes less than a minute.

It is worth emphasizing that \ins can be applied to any collection of
genomic elements, provided each one is sufficiently short that it does not
span regions having markedly different mutation rates or genealogies, and
provided each element can be associated with nearby sites likely to be free
from the effects of selection.  In this paper, we have focused on the case
of genome-wide collections of elements of a particular type, such as miRNAs
or binding sites for a particular transcription factor, but many other
types of analysis are possible.  For example, in related work
\citep{ARBIETAL12}, we have examined various subsets of TFBSs, such as
those associated with genes of a particular Gene Ontology category or
expressed at a various levels, and those having various levels of predicted
binding affinity.  As we have shown, the method can also be applied to
well-defined subsets of positions within elements, such as those
corresponding to particular motif positions or particular miRNA structural
regions.  Similar analyses could be used to contrast regions of the genome
having different epigenomic marks, sequences near to and far from genes,
sequences on sex chromosomes and autosomes, or any number of other
biologically significant genomic partitions.

\ins could be extended in various ways to improve the fit of the model to
the data and broaden the utility of the program.  In this analysis we had a
sufficiently large and complete collection of human variation data to
simply discard positions with missing data in one or more samples.  In
cases of more missing data, however, it may be worthwhile to use the
strategy of adjusting Watterson's constant $a_n$ in the appropriate
conditional distributions (see Table\ \ref{tab:cond-dist}) based on the
number of samples for which data is available at each genomic position.
This simple approach should work well as long as the amount of missing data is not
excessive, but it will require some care in programming to accommodate
site-wise variation in $a_n$ efficiently.  Another useful extension would
be to allow for variation across loci in the global parameters $\rho$,
$\eta$, and $\gamma$, say, by assuming locus-specific parameters are drawn
from Beta (for $\rho$) or Gamma (for $\eta$ and $\gamma$) distributions and
estimating the hyper-parameters for these distributions from the data.
This strategy should improve model fit considerably in cases of variable
selection across loci, similar to phylogenetic models that allow for rate
variation among sites \citep{YANG94}.  A further extension would be to use
a fully Bayesian approach and infer posterior distributions for the
parameters of interest.  This would also be fairly straightforward, but
would most likely require Markov chain Monte Carlo sampling or variation
Bayes approximations.  These and other extensions would help further in
using patterns of polymorphism and divergence to shed light on recent
evolutionary processes, particularly in noncoding regions, and may improve
predictions of the fitness effects of mutations across the genome.

\renewcommand*{\refname}{Literature Cited}

%


\captionsetup{list=no}

\clearpage{}
\section*{Tables}

\begin{table}[!h]
\caption{{\bf Model parameters}}
\noindent \begin{centering}
\begin{tabular}{lll}
\hline
{\bf Parameter} & {\bf Type} & {\bf Description}\\
\hline
$\vect{\lambda^O} = \{\lambda^O_b\}_{b\in B}$ & neutral & Block-specific neutral scaling factor for the outgroup portion of\\
&& the phylogeny, used when computing the prior distributions for\\
&& each deep ancestral alleles, $P(Z_i\ |\ O_i,\lambda^O_b)$\\
$\vect{\lambda} = \{\lambda_b\}_{b\in B}$ & neutral & Block-specific neutral scaling factor for divergence\\
$\vect{\theta} = \{\theta_b\}_{b\in B}$ & neutral & Block-specific neutral polymorphism rate\\
$\vect{\beta}=(\beta_1,\beta_2,\beta_3)$ & neutral & Relative frequencies
of the three derived allele frequency classes, $(0,f)$,\\
&& $[f,1-f]$, and $(1-f,1)$, within neutral polymorphic sites\\
$\rho$ & selection & Fraction of sites under selection within functional elements\\
$\eta$ & selection & Ratio of divergence rate at selected sites to local neutral\\
&& divergence rate\\
$\gamma$ & selection & Ratio of polymorphism rate at selected sites to local neutral\\
&& polymorphism rate\\
\hline
\end{tabular}
\par\end{centering}
\begin{flushleft}
\end{flushleft}
\label{tab:parameters}
\end{table}

\begin{table}[!h]
\caption{{\bf Model variables associated with site $\boldsymbol{i}$}}
\noindent \begin{centering}
\begin{tabular}{lll}
\\[-3ex]
\hline
{\bf Variable} & {\bf Type} & {\bf Description}\\
\hline
$O_{i}$ & observed & Set of aligned bases from outgroup species\\
$X_{i}^{\text{maj}}$ & observed & Base for major allele in target population\\
$X_{i}^{\text{min}}$ & observed & Base for minor allele in target population
(NA for monomorphic sites)\\
$Y_{i}$ & observed & MAF class for site $i$: 
`M' for monomorphic sites (MAF=0)\\
&&~~~~~~~~~~~~~~~~~~~~~~~~~~~~~~~~~~
`L' for polymorphic sites with MAF $<f$ \\
&&~~~~~~~~~~~~~~~~~~~~~~~~~~~~~~~~~~
`H' for polymorphic sites with MAF $\geq f$\\
$S_{i}$ & hidden & Selection class: `$\neut$' for neutral sites\\
&&~~~~~~~~~~~~~~~~~~~~~~~~~ `$\sel$' for sites under selection\\
$Z_{i}$ & hidden & Ancestral allele at the most recent common ancestor (MRCA)
of the target\\
&& population and the closest outgroup \\
$A_{i}$ & hidden & Ancestral allele at the MRCA of samples from the target
population\\ 
\hline
\end{tabular}
\par\end{centering}
\begin{flushleft}
\end{flushleft}
\label{tab:variables}
\end{table}

\begin{table}[!h]
\caption{{\bf Conditional distribution table for $\boldsymbol{P(X_i\ |\
      S_i,Z_i,\ \vect\zeta)}$}}
\label{tab:cond-dist}
\begin{center}
\begin{minipage}{5.5in}
\begin{center}
\begin{tabular}{llll}
\\[-5ex]
\hline
$s$ & $y$ & $z,~x_i^{\text{maj}},~x_i^{\text{min}}$~~[\footnote{Relationships
  among variables.  It is implicit that $x^{\text{maj}} \in$ \{A, C, G,
  T\} and $x^{\text{maj}} \ne x^{\text{min}}$ in all cases.  In addition,
  $x^{\text{min}}=\emptyset$ 
  when $y=\text{M}$}] & $P\left(X_i=(x^{\text{maj}},x^{\text{min}},y)\ |\ S_i=s, Z_i=z,\ \vect\zeta\right)$ \\
\hline
$\neut$ & M & $z = x^{\text{maj}}$ & $(1-\lambda_bt)(1-\theta_ba_n)$\\
$\neut$ & M & $z \neq x^{\text{maj}}$ & $\frac 13\lambda_bt(1-\theta_ba_n)$\\
$\neut$ & L & $z = x^{\text{maj}}$     & $\left((1-\lambda_bt)\beta_1 + \frac 13 \lambda_bt \beta_3\right)\frac13\theta_ba_n$\\
$\neut$ & L & $z = x^{\text{min}}$     & $\left((1-\lambda_bt)\beta_3 + \frac 13 \lambda_bt \beta_1\right)\frac13\theta_ba_n$\\
$\neut$ & L & $z\notin \{x^{\text{maj}},x^{\text{min}}\}$          & $\frac 13 \lambda_bt \left(\beta_1 + \beta_3 \right)\frac13\theta_ba_n$\\
$\neut$ & H & $z\in    \{x^{\text{maj}},x^{\text{min}}\}$          & $\left(1-\lambda_bt + \frac 13 \lambda_bt \right)\beta_2\frac13\theta_ba_n$\\
$\neut$ & H & $z\notin\{x^{\text{maj}},x^{\text{min}}\}$           & $\frac 23 \lambda_bt \beta_2\frac13\theta_ba_n$\\
$\sel$  & M & $z = x^{\text{maj}}$ & $(1-\eta\lambda_bt)(1-\gamma\theta_ba_n)$\\
$\sel$  & M & $z \neq x^{\text{maj}}$ & $\frac 13\eta\lambda_bt$\\
$\sel$  & L & $z = x^{\text{maj}}$     & $(1-\eta\lambda_bt)\frac 13\gamma\theta_ba_n$\\
$\sel$  & L & $z \neq x^{\text{maj}}$                          & $0$\\
$\sel$  & H &  ---                           & $0$\\
\hline
\end{tabular}
\par\end{center}
\end{minipage}
\end{center}
\end{table}

\clearpage{}
\section*{Figure Legends}

\setlength{\parindent}{0pt} 
\setlength{\parskip}{2ex}


{\bf Figure \ref{fig:model-schematic}. Schematic description of \ins.}
The method measures the influence of natural selection by contrasting
patterns of polymorphism and divergence in a collection of genomic elements
of interest (gold) with those in flanking neutral sites (dark gray).
Nucleotide sites in both elements ($E_b$) and flanks ($F_b$) are grouped
into a series of genomic blocks ($b$) to accommodate variation along the
genome in mutation rates and genealogical backgrounds.  The model consists
of phylogenetic (gray), recent divergence (blue), and intraspecies
polymorphism (red) components, which are applied to genome sequences for
the target population ($X$, red) and outgroup species ($O$, gray).  At each
nucleotide position, the alleles at the most recent common ancestors of the
samples from the target population ($A$) and of the target population and
closest outgroup ($Z$) are represented as hidden variables and treated
probabilistically during inference.  The allele $Z$ determines whether or
not monomorphic sites are considered to be divergent (D).  Polymorphic
sites are classified as having low- (L) or high- (H) frequency derived
alleles based on $A$ and a frequency threshold $f$.  The labels shown here
are based on a likely setting of $Z$ and $A$.  Vertical ticks represent
single nucleotide variants relative to an arbitrary reference.  Inference
is based on differences in the patterns of polymorphism and divergence
expected at neutral and selected sites.

\clearpage
{\bf Figure \ref{fig:graphical-model}.  Graphical model for a given
  nucleotide site $\boldsymbol{i}$.}
As in Fig.\ \ref{fig:model-schematic}, the phylogenetic portion of the
model is shown in gray, the divergence component in blue, and the
polymorphism component in red.  Observed variables are represented by solid
circles and hidden variables by empty circles.  The observed alleles in the
target population and outgroups are represented by $X_i$ and $O_i$,
respectively.  $X_i$ is further summarized using a major
($X_{i}^{\text{maj}}$) and minor ($X_{i}^{\text{min}}$) allele, as well as
the minor allele frequency class ($Y_i$; not shown).  The selection class
is denoted 
$S_i$, and the ancestral alleles are denoted $Z_{i}$ and $A_{i}$, as
described in Fig.\ \ref{fig:model-schematic}. Conditional dependence
between the variables is indicated by directed edges, in the standard
manner for probabilistic graphical models.  Model parameters are shown
alongside the associated conditional dependency edges.  The selection
parameters $\vect\zeta_{\text{\sel}}=\left(\rho,\eta,\gamma\right)$ are
highlighted in green.
\clearpage
{\bf Figure \ref{fig:sim-results}. Simulation results.}
{\bf (A)} Parameter estimates for four collections of 20,000 simulated elements
based on different mixtures of neutral (neut), strong positive (SP), strong
negative (SN), and weak negative (WN) selection (as indicated at bottom).
The true values of $\rho$, $\PD$, and $\WP$ are indicated by solid bars,
and estimates from \ins are indicated by diamonds, with error bars
representing one standard error.  For comparison, estimates from several
simpler count-based methods are also shown, including estimates of $\rho$
based on polymorphism ($\hat\rho_\text{Poly}$; `+') and divergence
($\hat\rho_\text{Div}$; solid squares) rates, and estimates of $\PD$ based
on the McDonald-Kreitman framework ($\hat\PD_{\text{-MK}}$; `$\times$').
Adaptive divergences ($\PD$) and deleterious polymorphisms
($\WP$) are shown as rates per 1,000 base pairs (kbp).  See {\bf Methods}
for details.
{\bf (B)} \ins was
applied to 11 collections of 10,000
elements with various fractions of sites under selection (see text),
assuming a range of values for the low-frequency derived allele threshold
$f$.  Shown are (left column) relative estimation errors for $\rho$, $\PD$, and
$\WP$, measured as differences between the estimates and true values
normalized by the true value, and (right column) curvature-based standard
errors (SE) for the estimates, both as a function of the frequency threshold
$f$.  Each boxplot describes the 
distribution of values for the 11 collections considered.
{\bf (C)} Each of the same 11 selection mixtures was combined with
four different demographic scenarios having varying degrees of
complexity (Supplementary Table \ref{tab:demography}).  Box plots represent
the distribution of relative error across the eleven collections for each
demographic scenario. The relative estimation error for the simple
site-count-based estimates, $\hat\rho_\text{Poly}$, $\hat\rho_\text{Div}$,
and $\hat\PD_{\text{-MK}}$ is shown for comparison.
%

\clearpage
{\bf Figure \ref{fig:genomic-results}. Analysis of human genomic elements.} 
{\bf (A)} Distribution of likelihood ratio test statistics for 500 sampled sets
of ``neutral" genomic elements, with $\sim$30,000 elements per set.  Test
statistics reflect a null hypothesis that $\rho=0$ and an alternative
hypothesis that $\rho>0$.  For comparison, a $\chi_3^2$ distribution (with
three degrees of freedom; red) and a 50:50 mixture of a $\chi_3^2$
distribution and a point mass at 0 (green) are also shown.  Blue lines
indicate significance thresholds for $p=0.01$ and $p=0.05$ based on the
$\chi_3^2$ distribution.  Four of the 500 data sets (0.8\%) had test
statistics exceeding the $p=0.01$ cutoff, and 24 (4.8\%) exceeded the
$p=0.05$ cutoff, indicating a reasonably good fit to the tail of the
distribution. The distribution of estimated values of $\rho$ is shown in
Supplementary Figure \ref{fig:neutral-rho-hist}.
{\bf (B)} Model-based estimates of $\rho$, $\Ea$, and $\Ew$ for three classes of
noncoding RNAs (lincRNAs, miRNAs, and snoRNAs), promoter regions, and GATA2
binding sites (see {\bf Methods}).  Error bars indicate one standard error.
For comparison, estimates of $\rho$ based on polymorphism
($\hat\rho_{\text{Poly}}$) and divergence ($\hat\rho_{\text{Div}}$) counts
are also shown.  Symbols in red indicate statistical significance in
likelihood ratio tests for overall selection ($\rho>0$; `*' $\rightarrow$
$p<0.01$), positive selection ($\eta>0$; `p' $\rightarrow$ $p<0.01$), and weak
negative selection ($\gamma>0$; `w' $\rightarrow$ $p<0.01$), based on a
$\chi_3^2$ distribution for $\rho>0$ and a $\chi_1^2$ distribution for
$\eta>0$ and $\gamma>0$.
{\bf (C)} The motif inferred for GATA2 together with position-specific estimates
of $\rho$ (left axis), $\PD$, and $\WP$ (right axis).  Statistical
significance is assessed and indicated as in (B).  The ``core'' seven
positions of the motif, having IC$>\frac 12$, are highlighted in gray.
Note that all seven core positions display significant evidence of
selection. In addition, positions 7 and 8 show significant evidence of
positive selection, and positions 5 and 6 show significant evidence of weak
negative selection.
{\bf (D)} Estimates of $\rho$ for several structural regions of miRNAs (inset).
(Left) Results for a coarse-grained partitioning into
loop bases, unpaired stem bases, and paired stem bases.  (Right) Results
for a finer-grained partitioning of paired
bases in the stem into loop-proximal, lower-stem,
star and mature regions, corresponding to the regions that undergo
cropping and dicing by Drosha and Dicer (dashed lines). Estimates found to
be significantly greater than 0 ($p\leq 0.01$) are highlighted (`*').


\clearpage{}
\section*{Figures}


\begin{figure}[h]
\begin{center}
\includegraphics[height=2.2in]{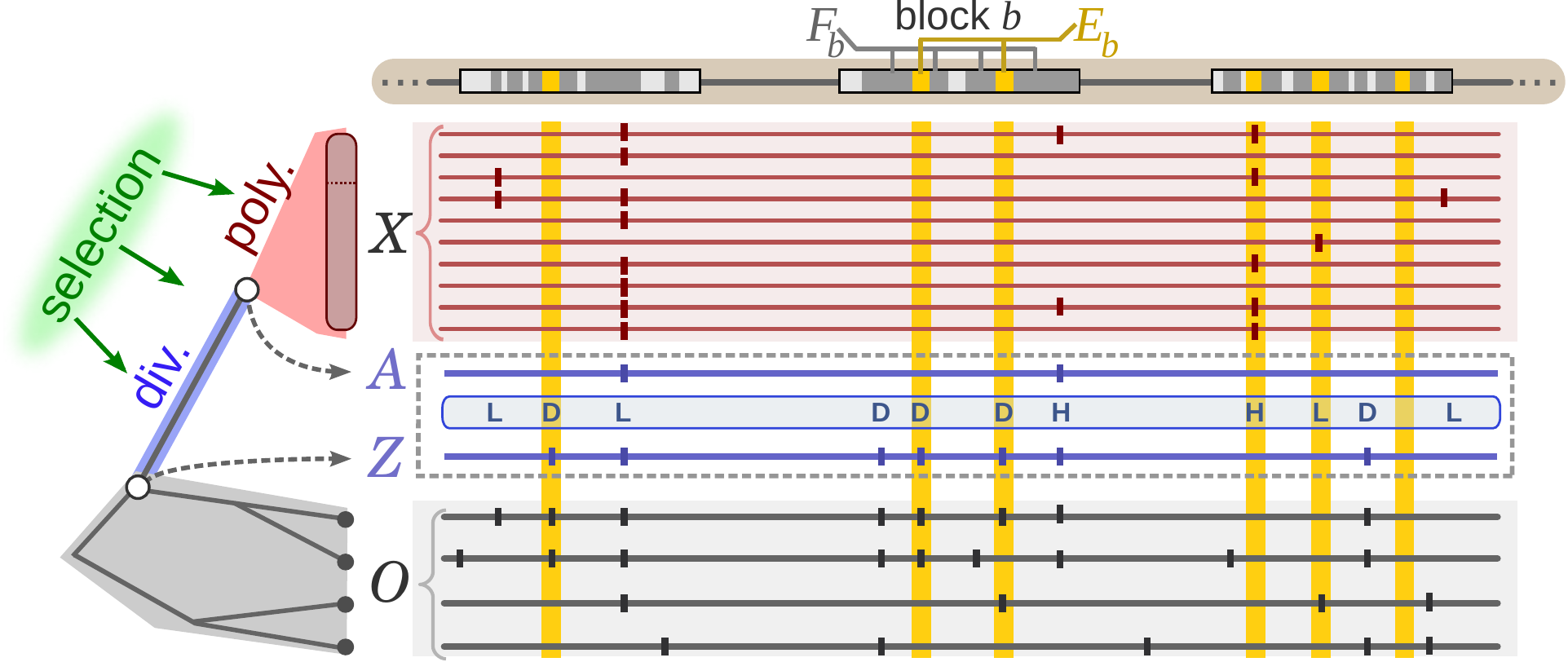}
\caption[Main Figure \ref{fig:model-schematic} -- \ins model schematic]{}
\label{fig:model-schematic}
\end{center}
\end{figure}
\vspace{60pt}
\begin{figure}[h!]
\begin{center}
\includegraphics[height=1.8in]{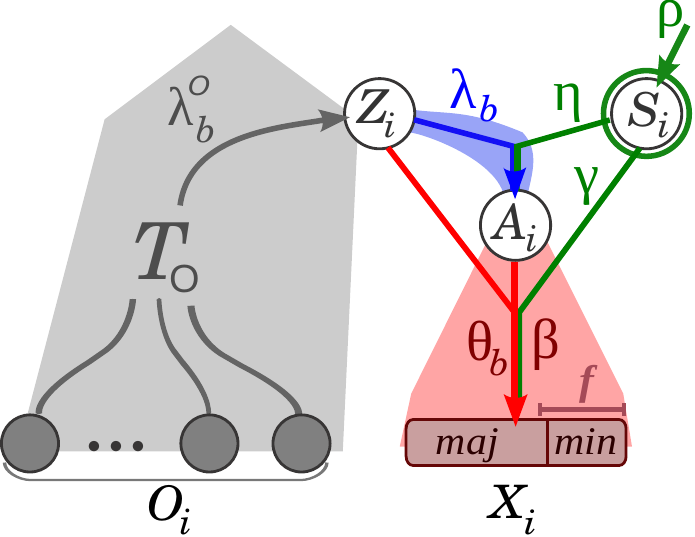}
%
%
\caption[Main Figure \ref{fig:graphical-model} -- \ins Graphical model]{}
\label{fig:graphical-model}
\end{center}
\end{figure}
\begin{figure}[h]
\begin{center}
\includegraphics[width=6.5in]{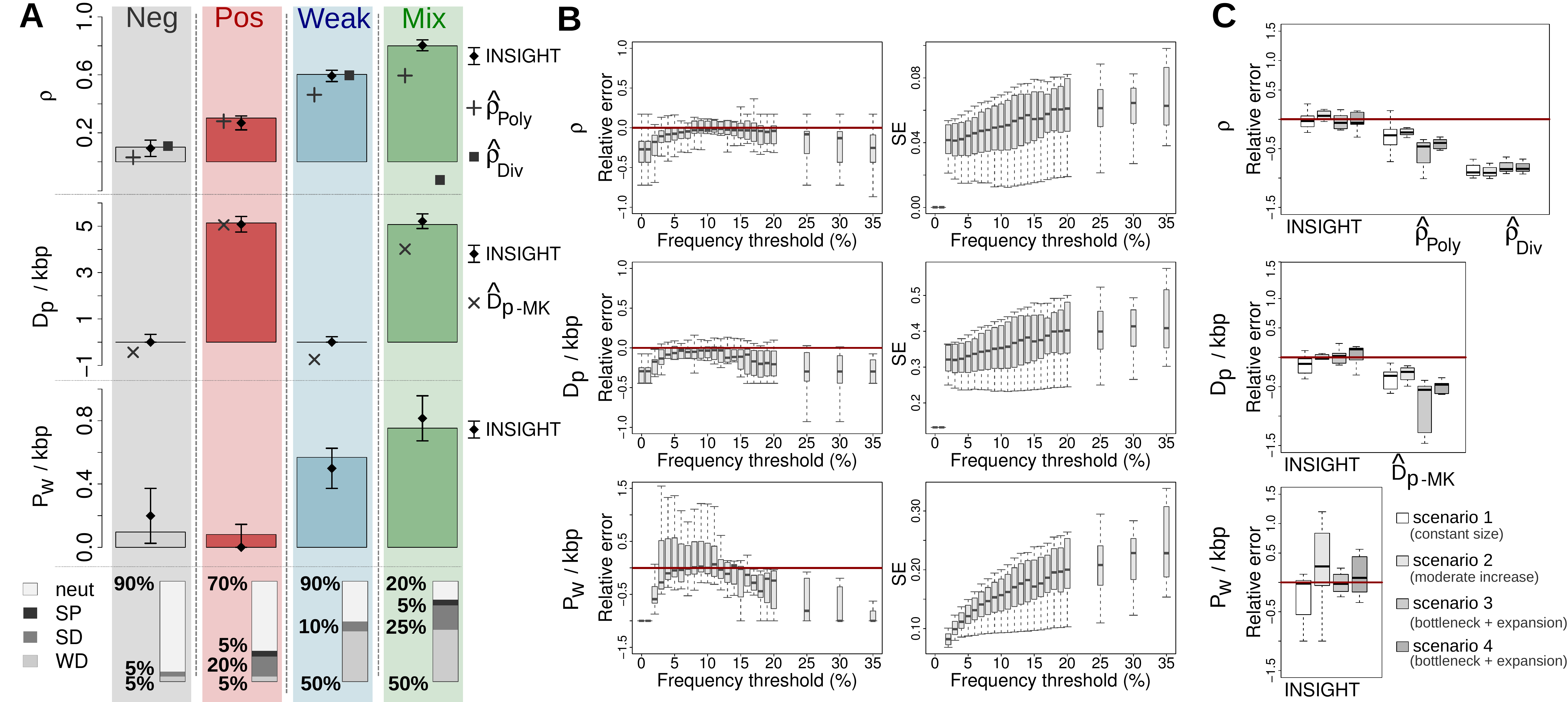}
\caption[Main Figure \ref{fig:sim-results} -- Simulation results]{}
\label{fig:sim-results}
\end{center}
\end{figure}
\begin{figure}[h]
\begin{center}
\includegraphics[width=6.5in]{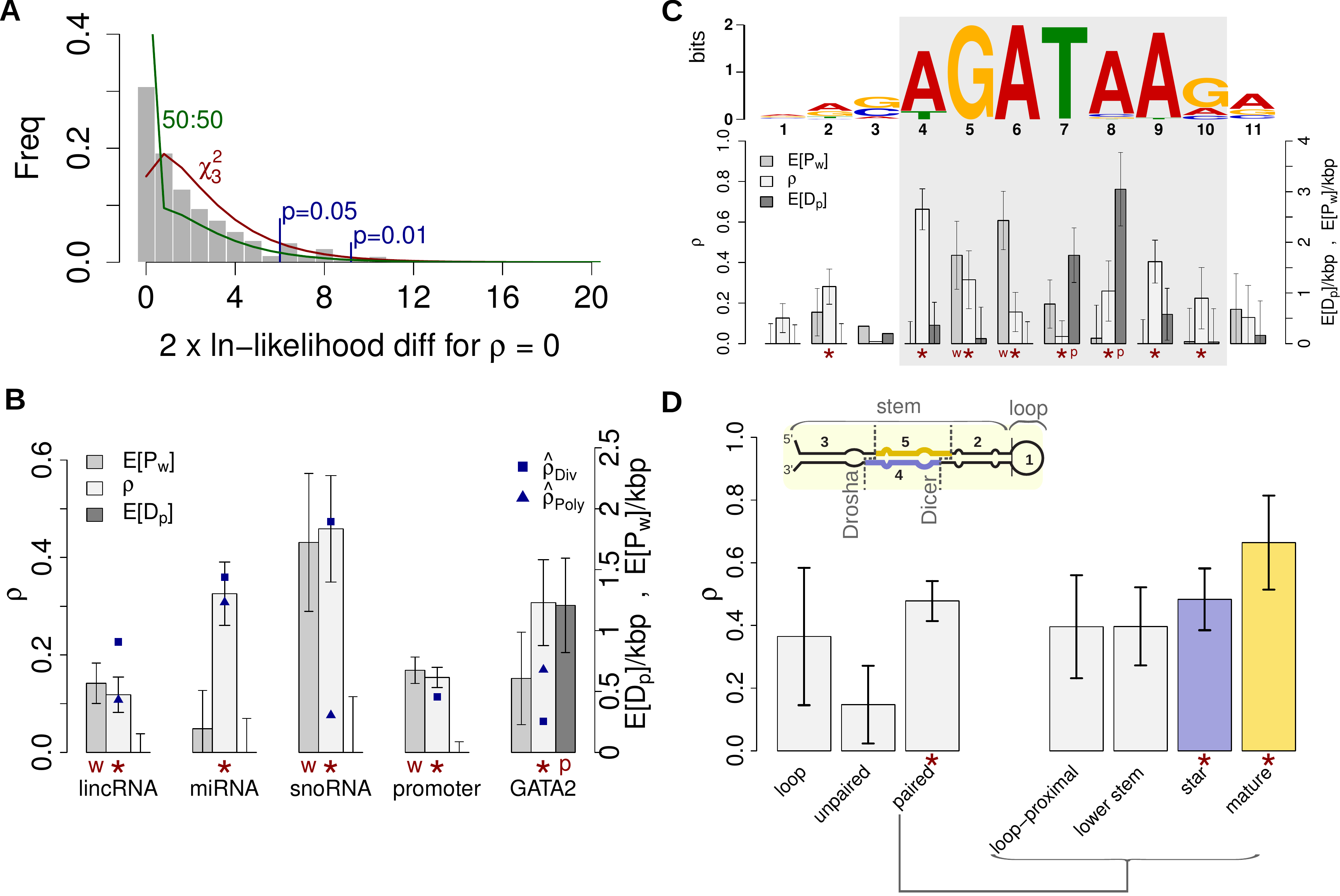}
\caption[Main Figure \ref{fig:genomic-results} -- Data analysis results]{}
\label{fig:genomic-results}
\end{center}
\end{figure}
\captionsetup{list=yes}
\clearpage

\clearpage

\pagestyle{myheadings}
\markright{Supplementary Material}

\oddsidemargin 0in
\evensidemargin 0in
\topmargin -.5in
\textwidth 6.5in
\textheight 9in


\makeatletter
\renewcommand{\@biblabel}[1]{\quad#1.}

\providecommand{\\}{\\}


\renewcommand{\thefigure}{S\arabic{figure}}
\renewcommand{\thetable}{S\arabic{table}}

\makeatother

\renewcommand{\thefootnote}{\textit{\alph{footnote}}}

\date{}



\setcounter{secnumdepth}{-1} 
\setcounter{figure}{0} 
\setcounter{table}{0}

\newcommand{\pXmaj}{{p_i^{maj}}}
\newcommand{\pXmin}{{p_i^{min}}}

\noindent \begin{LARGE}{\\ \bf \underline{Supplementary Material}}\end{LARGE}\\

\renewcommand\listfigurename{Supplementary Figures}
\listoffigures

\renewcommand\listtablename{Supplementary Tables}
\listoftables

\makeatletter
\def\tableofcontents{\section*{Supplementary Methods}\@starttoc{toc}}
\tableofcontents
\makeatother


\clearpage

\section*{Supplementary Tables}


\begin{table}[!h]
\caption[Demographic scenarios used in simulation]
{\bf Demographic scenarios used in simulation}
\noindent \begin{centering}
\begin{tabular}{lcccc}
\\
\hline
Time $^a$& scenario 1 $^b$ & scenario 2 $^c$ & scenario 3 $^d$ & scenario 4 $^d$\\
\hline
220 kya $^e$ & -- & 1.23x $^h$ & 1.23x $^h$ & 1.23x $^h$ \\
140 kya $^f$ & -- &     --   & 0.17x $^h$  & 0.17x $^h$ \\
20.8 kya $^g$ & --&     --   &  0.476x $^h$ & 0.242x $^h$ \\
20.8 kya $^g$ & --&     --   & exp(79.8) $^i$ & exp(109.7) $^i$ \\
\hline
\end{tabular}
\par\end{centering}
\begin{flushleft}
\begin{small}
$^a$ Time of change in $N_e$ (kya = 1,000 years ago). All demographic scenarios start with an ancestral
$N_e$ of 10,000 after divergence from chimpanzee, 6.5 mya. Demographic scenarios follow the model suggested
by \cite{GUTEETAL09}.\\
$^b$ Scenario 1 no demographic changes throughout history.\\
$^c$ Scenario 2 corresponds to the demographic history of an African population with a single
moderate population expansion.\\
$^d$ Scenario 3 \& 4 correspond to the demographic histories of a European and East Asian population, resp., each with a moderate population expansion followed by two population bottlenecks and an exponential expansion.\\
$^e$ Moderate population expansion in the African population ancestral to all current human populations.\\
$^f$ Divergence point of an ancestral Eurasian population from an ancestral African population
associated with a strict population bottleneck in the ancestral Eurasian population.\\
$^g$ Divergence of European and East Asian population associated with additional bottlenecks
in both ancestral populations followed by exponential expansion.\\
$^h$ Instantaneous population size increase or decrease by a given multiplicative factor.\\
$^i$ Exponential population size expansion at a given rate expressed as
$\log(N_e^\text{final} /N_e^\text{initial} ) / $time, where time is in units of $2N_e$ generations.

\end{small}
\end{flushleft}
\label{tab:demography}
\end{table}

\begin{table}[!h]
\caption[Classes of genomic elements]
{\bf Classes of genomic elements analyzed by \ins}
\noindent \begin{centering}
\begin{tabular}{lrrr}
\\
\hline
Class & elements$^a$ & element sites $^b$ & flanking sites$^c$\\
\hline
lincRNAs$^c$ &  3,362 & 323,284 & 1.6 Mb \\
miRNAs$^c$   &  1,323 &  63,543 & 2.5 Mb \\
snoRNAs$^c$  &    416 &  22,331 & 0.3 Mb \\
proximal promoters$^d$ &
           18,453 & 613,339 & 20.3 Mb \\
GATA2 binding sites$^f$ &
          39,535 & 209,065 & 109.0 Mb \\
\hline
\end{tabular}
\par\end{centering}
\begin{flushleft}
\begin{small}
$^a$ Number of distinct elements in class.\\
$^b$ Number of site in entire collection after filtering ($|E|$).\\
$^c$ Number of neutral flanking sites in megabases (Mb) used for neutral inference ($|F|$).\\
$^d$ ``Exon" level transcripts tagged as ``known" in GENCODE v.13.\\
$^e$ Proximal promoters are defined as the 100 bp region upstream the transcription start site
of known genes.\\
$^e$ Binding sites identified in ChIP-seq peaks from ENCODE data on multiple cell lines.
\end{small}
\end{flushleft}
\label{tab:element-classes}
\end{table}

\begin{table}[h]
\caption[Joint Probabilities Used in E Step of EM Algorithm]
{{\bf Joint Probabilities with Data Used in E Step of EM Algorithm}$^a$}
\begin{small}
\begin{tabular}{ll}
\hline
$Y_i$& $q(S_i=\neut)$ \\
\hline
$\text{M}$ & 
$(1-\rho^{(k)})\ \left( (1-\hat\lambda_bt)\pXmaj + \frac 13 \hat\lambda_bt(1-\pXmaj) \ \right)(1-\hat\theta_ba_n)$ \\
$\text{L}$ &
$(1-\rho^{(k)})\left[(1-\hat\lambda_bt)\ \ \left(\hat\beta_1\pXmaj + \hat\beta_3\pXmin\right) + \frac 13\hat\lambda_bt\ \left(\hat\beta_1(1-\pXmaj) + \hat\beta_3(1-\pXmin)\right)\right]\frac 13 \hat\theta_ba_n$ \\
$\text{H}$ & 
$(1-\rho^{(k)})\ \left[(1-\frac 23\hat\lambda_bt)(\pXmaj+\pXmin) ~+~   \frac 23 \hat\lambda_bt(1-\pXmaj-\pXmin) \right]\frac 13 \hat\beta_2\hat\theta_ba_n$ \\ 
\hline
\end{tabular}\vspace{10pt}
\noindent \begin{centering}
\begin{tabular}{llll}
\hline
$Y_i$& $q(S_i=\sel,Z_i=X_i^{maj})$ & & $q(S_i=\sel,Z_i\neq X_i^{maj})$\\
\hline
$\text{M}$ & 
$\rho^{(k)}\ \pXmaj\ (1 - \eta^{(k)} \hat\lambda_b t)\ (1-\gamma^{(k)}\hat\theta_ba_n) $ & &
$\frac 13 \rho^{(k)}\ (1-\pXmaj)\ \eta^{(k)} \hat\lambda_b t$\\
$\text{L}$ & 
$\rho^{(k)}\ \pXmaj\ (1 - \eta^{(k)} \hat\lambda_b t)(\frac 13  \hat\theta_b a_n\gamma^{(k)})\ $ & &
0 \\
$\text{H}$ & 
0 & &
0 \\ 
\hline
\end{tabular}
\par\end{centering}
\end{small}
\begin{flushleft}
\begin{small}
$^a$ The joint probabilities associated with site $i$, 
 $q(S_i=\neut)$, $q(S_i=\sel,Z_i=X_i^{maj})$, and $q(S_i=\sel,Z_i\neq X_i^{maj})$, are defined in
Equations \ref{eqn:qn} - \ref{eqn:qsd}. Values of the selection parameters in the 
$k$th iteration of the EM algorithm are represented by $\rho^{(k)}$, 
 $\eta^{(k)}$, and  $\gamma^{(k)}$, and the previously-estimated (and fixed) neutral model
parameters are represented by $\hat\beta_1, \hat\beta_2, \hat\beta_3, \hat\lambda_bt$, and
$\hat\theta_ba_n$ ($b$ is the genomic block that contains site $i$).
We use the following notation for the deep ancestral priors:
$\pXmaj \equiv P(Z_i=X_i^{maj}\ |\ O_i, \hat\lambda^O_b)$ ; $\pXmin \equiv P(Z_i=X_i^{min}\ |\ O_i, \hat\lambda^O_b)$
\end{small}
\end{flushleft}
\label{tab:em-posterior}
\end{table}

\clearpage

\section*{Supplementary Figures}

\begin{figure}[h!]
\begin{center}
\includegraphics[width=3in]{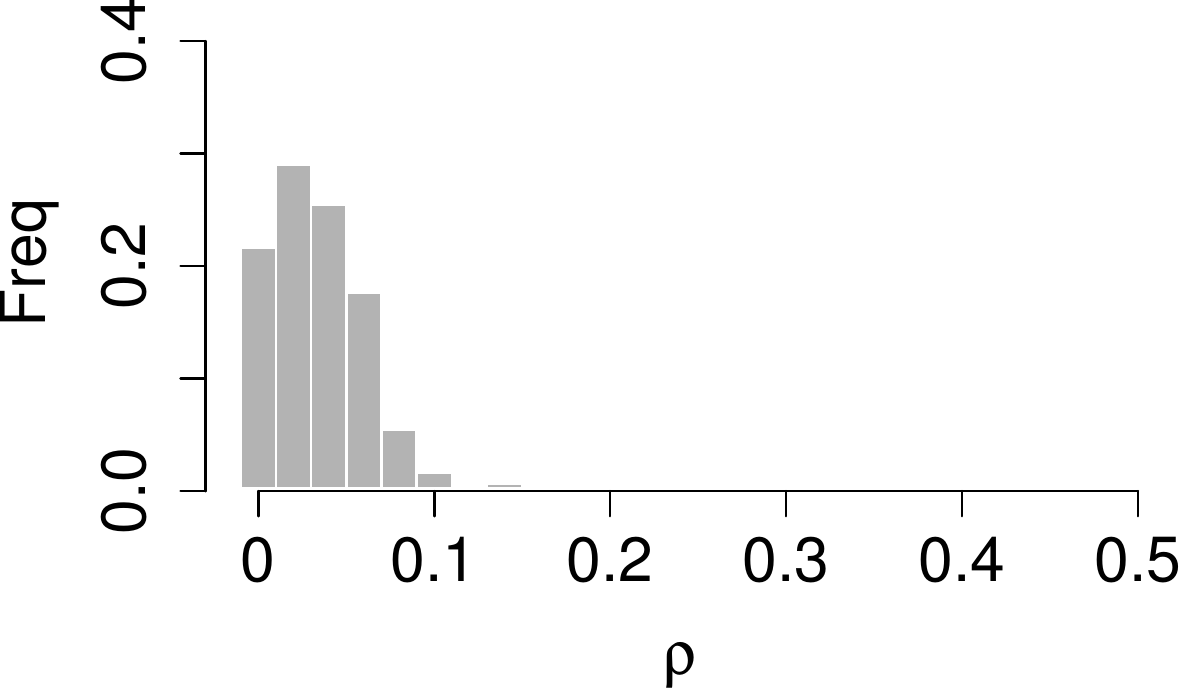}
\caption[Observed distribution of $\rho$ estimates
across 500 ``neutral" collections of elements]
{The observed distribution of $\rho$ estimates
for 500 collections of ``neutral" elements extracted from
putative neutral regions used in our control study (see Fig.\ \ref{fig:genomic-results}A).
Estimated values of $\rho$ had a median of 0.03 and a maximum of 0.17.
}
\label{fig:neutral-rho-hist}
\end{center}
\end{figure}

\clearpage

\begin{figure}[h!]
\begin{center}
\includegraphics[width=5in]{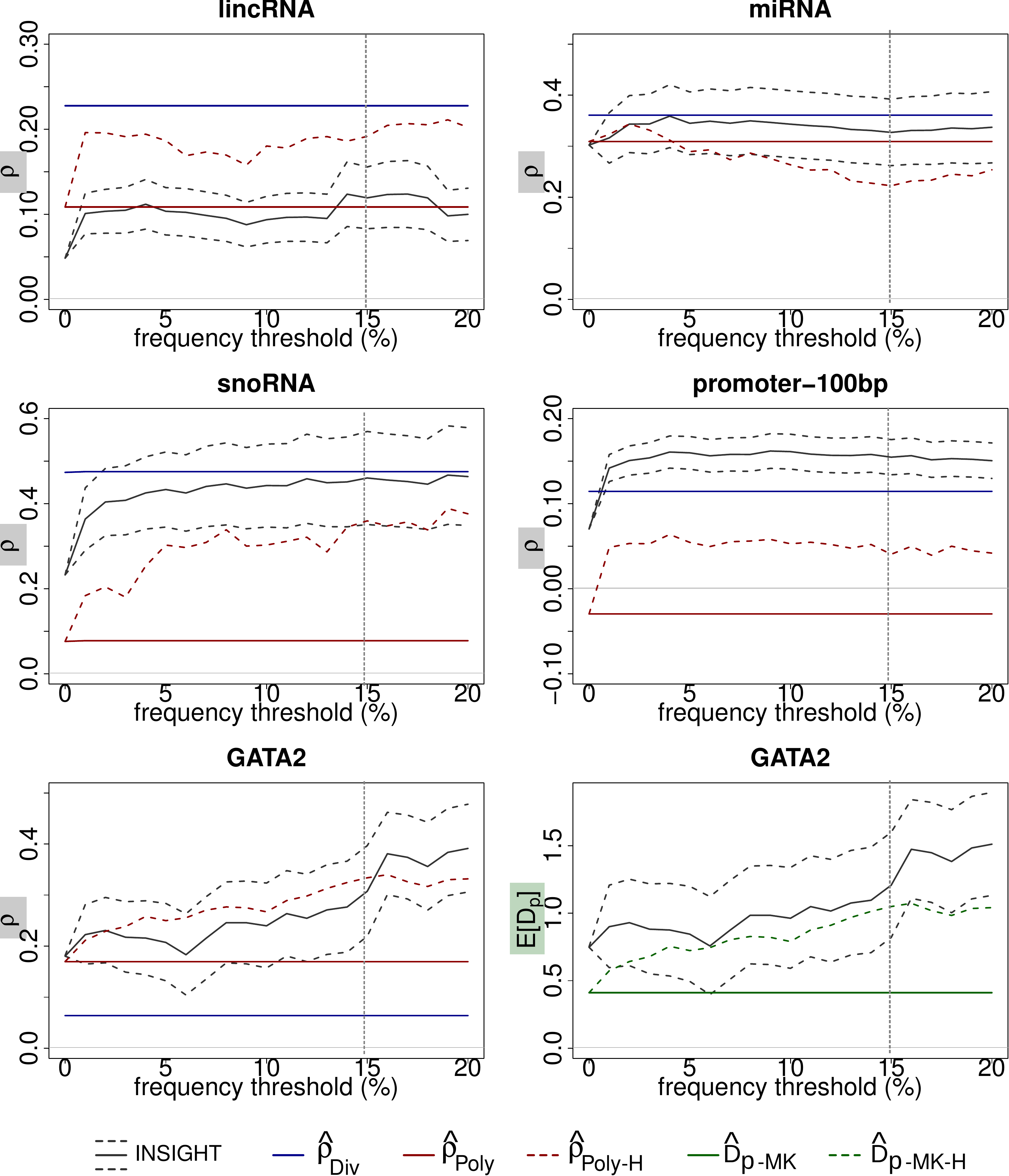}
\caption[Analysis of genomic elements using alternative frequency thresholds]
{Analysis of the five classes of human noncoding genomic elements for alternative
values of the frequency threshold $f$. The \ins estimates (solid lines) are shown
with their curvature-based confidence intervals (dashed lines). A frequency threshold of 15\%
(dotted vertical line) was used in our main analysis. Count-based
estimates for rates of divergent
 ($\hat\rho_{\text{Div}}$; horizontal blue line) and polymorphic 
($\hat\rho_{\text{Poly}}$; horizontal red line) sites
are shown for comparison. Corrected polymorphism-based estimates obtained by discarding
low-frequency polymorphisms ($\hat\rho_{\text{Poly-H}}$; Equation
\ref{eqn:rho-poly-h}; dashed red line) are shown as well.
In addition, estimated posterior expected numbers of adaptive divergences ($\Ea$) are shown for GATA2
(bottom right) as a function of the frequency threshold, together with the MK-based
estimates ($\hat\PD_{\text{-MK}}$; horizontal green) and a corrected MK-based estimates
excluding low-frequency polymorphisms ($\hat\PD_{\text{-MK-H}}$; Equation \ref{eqn:mk-pd-h}; dashed green).
}
\label{fig:genomic-thresholds}
\end{center}
\end{figure}

\clearpage

%
%

%
%

\begin{figure}[h!]
\begin{center}
\includegraphics[width=4in]{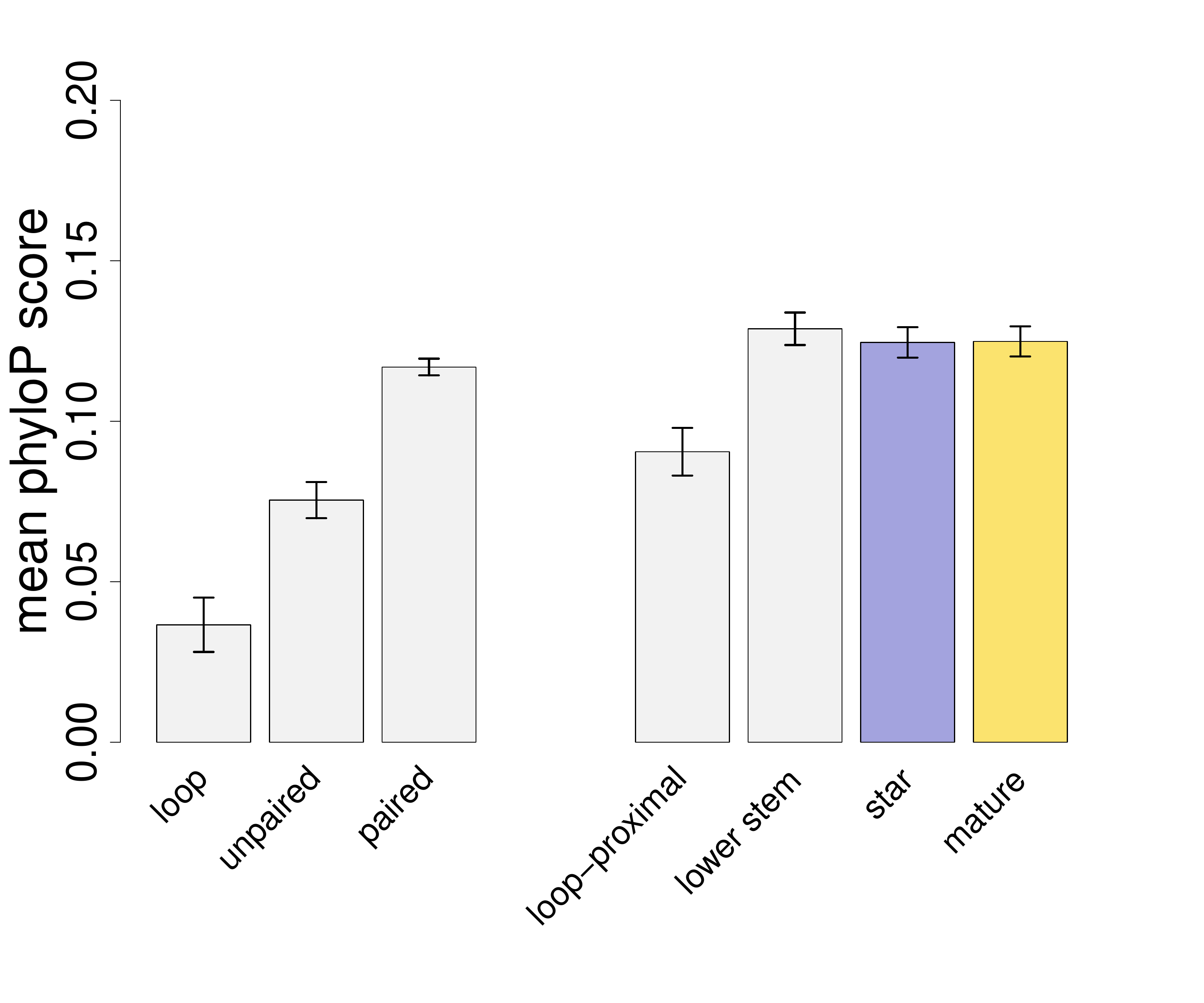}
\caption[Mean conservation scores for five miRNA components]
{Sequence conservation for the five miRNA considered in our analysis
(Fig.\ \ref{fig:genomic-results}D) estimated from a multiple sequence
alignment of eleven primate 
species using phyloP \citep{POLLETAL09}.  The plot represents mean sitewise
phyloP scores with standard errors.
The main trends are consistent with our \ins-based estimates, namely,
higher conservation in paired vs. unpaired bases, and lower conservation in the loop region vs.
the stem region. However, unlike the conservation scores, the \ins-based
estimates of $\rho$ 
indicate stronger selection for the mature strand than for other regions of the
miRNA stem.
}
\label{fig:mirna-phylop}
\end{center}
\end{figure}

\clearpage

\section*{Supplementary Methods}

\subsection{Detailed Inference Algorithm}
\label{supp:inference}

\newcommand{\et}{{e_{\text{x}}}}


Recall that the probabilistic model of \ins has four global parameters (see Table \ref{tab:parameters}):
the fraction of sites under selection in elements ($\rho$), the relative
divergence ($\eta$) and polymorphism ($\gamma$) rates at selected sites,
and $\vect\beta$, a multivariate parameter summarizing the neutral site
frequency spectrum. In addition, each genomic block, $b\in B$ is associated with
three block-specific paramters:
a population-scaled mutation rate ($\theta_b$), a neutral divergence scale
factor ($\lambda_b$), and an outgroup divergence scale factor
($\lambda^O_b$). The objective of the inference algorithm is to estimate approximate MLEs
for all these parameters given the following input:
\begin{enumerate}
\item \label{it:poly} The polymorphism class $Y_i$ and the major and minor alleles $(X_i^{maj},X_i^{min})$
observed across the sampled individuals in the target population, for each site $i$ along the genome.
\item \label{it:out} An alignment of each outgroup species to the reference sequence of the target population (the columns of this multiple genome alignment are represented by $\{O_i\}$ in our graphical model).\item \label{it:tree} A phylogenetic tree $T$ with branch lengths and an instantaneous substitution rate matrix. In the tree $T$, we denote by $\et$ (or {\em target branch}) the terminal branch in $T$ leading to the target population and by $T_{\text{O}}$
the outgroup phylogeny consisting of all branches in $T$ other than $\et$  (see Fig.\ \ref{fig:model-schematic}A).

\item \label{it:flanks} A collection of putative neutral sites $F$.
\item \label{it:elements} A collection of sites within functional elements $E$.
\item \label{it:blocks} A partitioning of the genome into a series of mutually exclusive blocks $B$.
\end{enumerate}

The inference procedure consists of three separate stages: (1) phylogenetic
model fitting, (2) neutral polymorphism model fitting, and (3) selection
inference. The neutral inference, which consists of the first two stages,
makes use of all the above input components other than the collection of
functional element sites $E$. The selection inference stage uses the output
of the first two stages together with the polymorphism data across sites in
$E$. This implies that the neutral inference can be executed without any
knowledge of $E$.  Its output can be stored and later 
contrasted against multiple
collections of functional elements (using a single application of the
selection inference stage for each such collection).  

\subsubsection{Phylogenetic Model Fitting}\label{subsec:phylogenetic-fitting}

The phylogenetic model fitting stage makes use of all input components mentioned above other than (\ref{it:elements}), the collection of element sites.
The phylogenetic model is fitted separately to each genomic block to account for variation in mutation rate
along the genome. For each block $b$, we fit two scaling factors for branch lengths in the phylogenetic tree $T$:
a divergence scaling factor ($\hat\lambda_b$) corresponding to the target branch, $\et$,
and an outgroup scaling factor ($\hat\lambda^O_b$) corresponding to the outgroup
portion of the phylogeny, $T_{\text{O}}$. The two scaling factors are fitted to the outgroup alignments with a single call
to the {\em phyloFit} function from RPHAST \citep{HUBIETAL11}, using a user-specified instantaneous substitution rate
matrix. Both scaling factors are fitted to the multiple sequence alignment of outgroup genomes to the reference sequence
at the putative neutral sites ($F_b$) within block $b$.
Note that obtaining accurate MLEs for all block-specific $\lambda_b$ parameters requires joint estimation with the
global parameters $\beta_1$ and $\beta_3$ (see below). However, since $\beta_1$ and $\beta_3$ affect the likelihood
only at polymorphic sites (specifically with $Y-i=\text{L}$), and those are
typically rare (fewer than 1\% of sites),
a good approximation of the MLEs can be obtained by fitting $\lambda_b$ while considering only monomorphic sites.
This allows us to separately infer each block-specific $\lambda_b$ and to decouple this task from that of inferring
$\beta_1$ and $\beta_3$.

After fitting the two scaling factors, the scaled outgroup phylogeny,
$\hat\lambda^O_b\cdot T_{\text{O}}$, is used to obtain the prior
distributions for the deep ancestral allele, $P(Z_i\ |\
O_i,\hat\lambda_b)$, for each site within block $b$ (not only the sites in
$F_b$). This computation is done by masking out the entire reference
sequence of the target population within the multiple alignment of
outgroups, and by applying the {\em postprob.msa} function in RPHAST to
compute the posterior probability distribution over bases at the ancestral
node at the root of $\et$ for each site along the alignment. It is worth
noting that although this is considered a ``posterior probability'' by
RPHAST, because the reference sequence is masked it is actually the
conditional probability of $Z_i$ given the outgroup genomes only, and
therefore can be considered as a conditional prior distribution in our
model.  The output of the phylogenetic model fitting 
stage is the set of estimated divergence rates, $\{\hat\lambda_bt\}$ ($t$
is the length of $\et$), and the deep ancestral priors across all
non-filtered sites in the genome, which we will denote from this point on
by $\{p(Z_i) \equiv P(Z_i\ |\ O_i,\hat\lambda_b )\}$. The estimates of the
outgroup scaling factors, $\{\hat\lambda^O_b\}$, are not used by any of the
subsequent stages.  Blocks with too few informative sites (100 bp in our
data analysis) are discarded, and element sites within these blocks are
filtered from the analysis.

\subsubsection{Neutral Polymorphism Model Fitting}\label{subsec:neutral-inference}

The neutral polymorphism model fitting stage makes use of input components (\ref{it:poly}),
(\ref{it:flanks}), and (\ref{it:blocks}) mentioned above, as well as the divergence rates,
$\{\hat\lambda_bt\}$, and deep ancestral priors, $\{p(Z_i) \}_{i\in F}$, computed in the previous stage.
This stage estimates the neutral polymorphism parameters $\{\theta_b\}$ and $\vect{\beta}=(\beta_1,\beta_2,\beta_3)$
by maximum likelihood, considering only sequence data in $F$ and conditioning on the previously estimated
divergence scales and ancestral priors. Assuming completely observed model variables, The log-likelihood function
can be expressed as a function of simple site category counts. For this purpose, we use the notation $c_Q(\mathcal{X})$
to indicate the number of sites within subset $Q$ with variable configuration $\mathcal{X}$. Using this notation,
the relevant portion of the likelihood function (Equation \ref{eqn:likelihood}; main text) can be expressed as follows:
\begin{small}
\begin{align}\label{eqn:neut-likelihood}
 \ln\ \left(\ld_F\left(\{\theta_b\},\vect{\beta}\ ;\ \vect{X}_F, \vect{O}_F, \vect{\hat\lambda}\right)\ \right)
& =~ \sum_{b=1}^B\left(c_{F_b}\left(Y_i\neq\text{M}\right) \ln(\theta_b) ~+
c_{F_b}(Y_i=\text{M})\ln(1-\theta_b a_n)
 ~\right)\\
\nonumber
&+~
c_F\left(Y_i=\text{L}, A_{i}=X^{maj}_{i}\right)\ln(\beta_1) ~+~ c_F\left(Y_i=\text{L}, A_{i}=X^{min}_{i}\right)\ln(\beta_3)\\
\nonumber &+~ c_F\left(Y_i=\text{H}\right)\ln(\beta_2)\\
\nonumber &+~~~C~,
\end{align}
\end{small}

\vspace{-5ex}
\noindent
where $C$ represents a term that does not depend on $\beta_1, \beta_2, \beta_3$, or $\{\theta_b\}$. 

Maximum likelihood estimates (MLEs) for $\{\theta_b\}$ and $\beta_2$ depend only on the observed variable
$Y_i$ and can thus be obtained using the following closed-form solutions:
\begin{align}
\label{eqn:ml-theta}\hat\theta_b &= \frac 1{a_n} \ \frac{ c_{F_b}\left(Y_i\neq\text{M}\right) }{|F_b|}~.\\
\label{eqn:ml-beta2}\hat\beta_2 &= \frac{ c_F\left(Y_i=\text{H}\right) }{c_F\left(Y_i\neq\text{M}\right)}
\end{align}
The MLEs for $\beta_1$ and $\beta_3$ depend on counts associated with the hidden variable $A_i$ and are estimated using the following EM algorithm.

\subsubsection*{EM Algorithm for $\boldsymbol{\beta_1}$ and $\boldsymbol{\beta_3}$:}
\begin{description}
\item[Initialization:] initialize the iteration counter $k\leftarrow 0$ and initialize $\beta_1$ and $\beta_3$ as follows: $\beta^{(0)}_1\leftarrow \delta(1-\hat\beta_2)$ and   $\beta^{(0)}_3\leftarrow (1-\delta)(1-\hat\beta_2)$, where $0<\delta<\frac 12$.
\item[Iterate until convergence:] ~
\begin{description}
\item[Expectation:] 
For every site $i\in F$ with $Y_i=\text{L}$, compute the following posterior probability:
\begin{small}
\begin{align}
\label{eqn:qmaj}
p(A_i=X^{maj}_i) &= ~P\left(A_{i}=X^{maj}_{i}\ |\ X_i,\ S_{i}=\text{\neut}, O_{i},\ \hat\theta_b,\hat\lambda_b,\hat\lambda^O_b,\beta^{(k)}_1\right) \\
\nonumber 
&=~\frac{\beta^{(k)}_1\psi\left(p(Z_{i}=X^{maj}_{i}),\hat\lambda_bt\right)}
{\beta^{(k)}_1\psi\left(p(Z_{i}=X^{maj}_{i}),\hat\lambda_bt\right) + \beta^{(k)}_3\psi\left(p(Z_{i}=X^{min}_{i}),\hat\lambda_bt\right)}
~, 
\end{align}
\end{small}

\vspace{-5ex}
\noindent
where $\psi(x,\hat\lambda_bt)=x(1-\hat\lambda_bt) +(1-x)\frac{1}{3}\hat\lambda_bt$
(see Table\ \ref{tab:cond-dist}).
\comment{

OLD POSTERIORS

\begin{align}
\label{eqn:qmaj}
\pmaj &=~ P\left(A_{i}=X^{maj}_{i},\ Y_{i}=\text{L} , \ X^{maj}_{i}, X^{min}_{i}\ |\ S_{i}=\text{\neut}, O_{i},\ \hat\theta_b,\hat\lambda_b,\hat\lambda^O_b,\beta^{(k)}_1\right) \\
\nonumber &=~  \frac 13 \hat\theta_ba_{i}\beta^{(k)}_1\left(p(Z_{i}=X^{maj}_{i})(1-\hat\lambda_bt) +p(Z_{i}\neq X^{maj}_{i})\frac{1}{3}\hat\lambda_bt \right)~. \\
\label{eqn:qmin}
\pmin &=~ P\left(A_{i}=X^{min}_{i},\ Y_{i}=\text{L} , \ X^{maj}_{i}, X^{min}_{i}\ |\ S_{i}=\text{\neut}, O_{i},\ \hat\theta_b,\hat\lambda_b,\hat\lambda^O_b,\beta^{(k)}_3\right)\\
\nonumber &=~  \frac 13 \hat\theta_ba_{i}\beta^{(k)}_3\left(p(Z_{i}=X^{min}_{i})(1-\hat\lambda_bt) +p(Z_{i}\neq X^{min}_{i})\frac{1}{3}\hat\lambda_bt) \right) ~.
\end{align}
\begin{equation}\label{eqn:beta1-exp}
\langle\ c_F(Y_i=\text{L}, A_{i}=X^{maj}_{i})\ \rangle ~=~ \sum_{i\in F\ |\ Y_i=\text{L}}\ \frac{\pmaj}{g+\pmin} ~.
\end{equation}

}

Then use these values to compute the following expected count:
\begin{equation}\label{eqn:beta1-exp}
\langle\ c_F(Y_i=\text{L}, A_{i}=X^{maj}_{i})\ \rangle ~=~ \sum_{i\in F\ |\ Y_i=\text{L}}\ p(A_i=X^{maj}_i) ~.
\end{equation}
\item[Maximization:] maximize the expected log-likelihood function by updating $\beta_1$ and $\beta_3$ as follows:
\begin{eqnarray}
\label{eqn:beta1-max} \beta^{(k+1)}_1 &=& \frac{ \langle\ c_F(Y_i=\text{L}, A_{i}=X^{maj}_{i})\ \rangle }{c_F\left(Y_i\neq\text{M}\right)}~,\\
\label{eqn:beta3-max} \beta^{(k+1)}_3 &=& 1-\hat\beta_2-\beta^{(k+1)}_1~.
\end{eqnarray}
\end{description}
\end{description}

Note that throughout the EM algorithm, the sum of $\beta_1$ and $\beta_3$ remains constant at $ 1-\hat\beta_2$. Also note that the computation in each iteration 
depends on the current values of $\beta_1$ and $\beta_3$, as well as on the
divergence rate parameters
$\{\hat\lambda_b\}$ and deep ancestral state priors $\{p(Z_{i})\}_{i\in F}$ computed in the
phylogenetic model fitting stage. It is, however, independent of 
the estimated polymorphism rate parameters $\{\hat\theta_b\}$, implying that $\vect{\beta}$ and $\{\theta_b\}$ can be estimated in parallel.

\subsubsection{Selection Inference}\label{subsec:selection-inference}

The selection inference stage receives as input the polymorphism data across element sites, $E$, as
well as the ancestral priors and neutral parameter estimates from the previous stages of inference.
This stage estimates the selection parameters $\rho, \eta$, and $\gamma$ by maximizing the log-likelihood function,
conditional on the previously estimated values of the neutral parameters. Assuming completely observed model variables,
The log-likelihood function can be expressed as a function of simple site category counts as follows (see also Equation
\ref{eqn:sel-likelihood}):
\begin{align}\label{eqn:sel-liklihood1}
\ln[\ld(&\rho,\eta,\gamma \ ; \data,\outgroup, \vect{\hat{\zeta}_\neut})]
 =\\
\notag &c_E(S_i=\text{\sel})\ln(\rho) ~+~c_E(S_i=\text{\neut})\ln(1-\rho) ~~+\\
\notag &c_E(S_i=\text{\sel},Z_i\neq X_i^{maj})\ln(\eta) ~+~ \sum_{b\in B}c_{E_b}(S_i=\text{\sel},Z_i= X_i^{maj})\ln(1-\eta\lambda_bt) ~~+ \\
\notag &c_E(S_i=\text{\sel},Y_i=\text{L})\ln(\gamma) ~+~ \sum_{b\in B}c_{E_b}(S_i=\text{\sel},Y_i=\text{M},Z_i= X_i^{maj})\ln(1-\gamma\theta_ba_n) ~~+~~C~,
\end{align}
where $C$ represents a term that does not depend on $\rho,\eta$, or $\gamma$.

Since the selection class $S_i$ is unknown, we use an EM algorithm, as detailed below, to find the MLEs for $\rho, \eta$, and $\gamma$.
Note that in the presence of missing data, $a_n$ varies along the genome, and the
second sum over blocks in Equation \ref{eqn:sel-liklihood1} is broken up into a sum over sites. The same
EM algorithm is applicable in this case, but the function maximized in each update of $\gamma$
would have more terms, and would thus take more time to numerically maximize (see details below).

\subsubsection*{EM Algorithm for $\boldsymbol{\rho, \eta}$ and $\boldsymbol{\gamma}$:}
\begin{description}
\item[Initialization:] initialize the iteration counter $k\leftarrow 0$ and
  select plausible initial values for $\rho, \eta$, and $\gamma$, for instance, $\rho^{(0)}\leftarrow 0.6, ~\eta^{(0)}\leftarrow 1.0$, ~and $\gamma^{(0)}\leftarrow 0.5$.
\item[Iterate until convergence:] ~
\begin{description}
\item[Expectation:] 
For every site $i\in E$, compute the following posterior distributions (see Table\ \ref{tab:em-posterior}):
\begin{align}
\label{eqn:pn}
&\pn &=&~~P(S_{i}=\text{\neut}\ |\ X_{i}, O_{i},\ \vect{\hat\zeta_{\text{\neut}}},\rho^{(k)},\eta^{(k)},\gamma^{(k)})~, \\
\label{eqn:psnd}
&\psnd &=&~~P(S_{i}=\text{\sel}, Z_{i}= X^{maj}_{i}\ |\ X_{i}, O_{i},\ \vect{\hat\zeta_{\text{\neut}}},\rho^{(k)},\eta^{(k)},\gamma^{(k)})~, \\
\label{eqn:psd}
&\psd &=&~~P(S_{i}=\text{\sel}, Z_{i}\neq X^{maj}_{i}\ |\ X_{i}, O_{i},\ \vect{\hat\zeta_{\text{\neut}}},\rho^{(k)},\eta^{(k)},\gamma^{(k)})~,
\end{align}
%
%
\noindent Then compute the following global expected counts:
\begin{eqnarray}
\label{eqn:rho-exp}
\langle\ c_E(S_i=\text{\sel})\ \rangle &=& \sum_{i\in E}\ \left(1-\pn\right) ~,\\
\label{eqn:eta-exp}
\langle\ c_E(S_i=\text{\sel},Z_i\neq X_i^{maj})\ \rangle &=& \sum_{i\in E}\ \psd ~,\\
\label{eqn:gamma-exp}
\langle\ c_E(S_i=\text{\sel},Y_i=\text{L})\ \rangle &=& \sum_{i\in E\ ,\ Y_i=\text{L}}\ \left(1-\pn\right) ~,
\end{eqnarray}
and the following block-specific counts:
\begin{eqnarray}
\label{eqn:eta-exp1}
\langle\ c_{E_b}(S_i=\text{\sel},Z_i= X_i^{maj})\ \rangle &=& \sum_{i\in E_b}\ \psnd ~,\\
\label{eqn:gamma-exp1}
\langle\ c_{E_b}(S_i=\text{\sel},Y_i=\text{M},Z_i= X_i^{maj})\ \rangle &=& \sum_{i\in E_b\ ,\ Y_i=\text{M}}\ \psnd ~.
\end{eqnarray}
\item[Maximization:] maximize the expected log-likelihood by updating $\rho$ as follows:
\begin{eqnarray}
\label{eqn:rho-max} \rho^{(k+1)} &=& \frac{ \langle\ c_E(S_i=\text{\sel})\ \rangle }{|E|}~,
\end{eqnarray}
and updating $\eta$ and  $\gamma$ by {numerically} finding the maxima for the two following functions, respectively:
\begin{small}
\begin{eqnarray}
\label{eqn:eta-max} f_1(\eta) &=& \langle\ c_E(S_i=\text{\sel},Z_i\neq X_i^{maj})\ \rangle\ln(\eta)  ~+~ \sum_{b}\langle\ c_{E_b}(S_i=\text{\sel},Z_i= X_i^{maj})\ \rangle \ln(1-\eta\hat\lambda_bt)~,\\
\label{eqn:gamma-max} f_2(\gamma) &=& \langle\ c_E(S_i=\text{\sel},Y_i=\text{L})\ \rangle\ln(\gamma) ~+~ \sum_{b}\langle\ c_{E_b}(S_i=\text{\sel},Y_i=\text{M},Z_i= X_i^{maj})\ \rangle\ln(1-\gamma\hat\theta_ba_n)~.~~~~~~~~~~
\end{eqnarray}
\end{small}
\end{description}
\end{description}

Due to variation in divergence and polymorphism rates across genomic blocks, the M-step updates for $\eta$ and $\gamma$ require numerical optimization. This optimization procedure uses standard
techniques for optimization of convex functions (see {\bf Procedure for Numerical Optimization} below).

The description of the EM algorithm is finalized by presenting the formulas used for the computation of
three posterior distributions, $\pn$, $\psnd$, and $\psd$, used in the E step of each iteration
(Equations \ref{eqn:pn}--\ref{eqn:psd}). In order
to compute these posteriors, we compute the joint distribution with the data for each of these three
variable configurations, and then normalize these joint distributions in the appropriate way. We define the
following notation for these joint distributions:
\begin{small}
\begin{align}
\label{eqn:qn}
&q(S_i=\neut) &=&~~P(S_{i}=\text{\neut}\ ,\ X_{i}, O_{i},\ \vect{\hat\zeta_{\text{\neut}}},\rho^{(k)},\eta^{(k)},\gamma^{(k)})~. \\
\label{eqn:qsnd}
&q(S_i=\sel,Z_i = X_i^{maj}) &=&~~P(S_{i}=\text{\sel}, Z_{i}= X^{maj}_{i}\ ,\ X_{i}, O_{i},\ \vect{\hat\zeta_{\text{\neut}}},\rho^{(k)},\eta^{(k)},\gamma^{(k)})~. \\
\label{eqn:qsd}
&q(S_i=\sel,Z_i \neq X_i^{maj}) &=&~~P(S_{i}=\text{\sel}, Z_{i}\neq X^{maj}_{i}\ ,\ X_{i}, O_{i},\ \vect{\hat\zeta_{\text{\neut}}},\rho^{(k)},\eta^{(k)},\gamma^{(k)})~.\\
\notag
&q_{total} &=&~~ P(X_{i}, O_{i},\ \vect{\hat\zeta_{\text{\neut}}},\rho^{(k)},\eta^{(k)},\gamma^{(k)})\\
&&=&~~q(S_i=\neut) ~+~ q(S_i=\sel,Z_i = X_i^{maj}) ~+~ q(S_i=\sel,Z_i \neq X_i^{maj})~.
\label{eqn:qTotal}
\end{align}
\end{small}

Expressions for $q(S_i=\neut)$, $q(S_i=\sel,Z_i = X_i^{maj})$, and $q(S_i=\sel,Z_i \neq X_i^{maj})$
are given in Table\ \ref{tab:em-posterior}, and each of the three posterior distributions, $\pn$, $\psnd$,
and $\psd$, is obtained by
normalizing the appropriate joint probability by the total probability associated with site $i$:
$p(\chi) = q(\chi)/q_{total}$.

\subsubsection{Procedure for Numerical Optimization of ``Sum of Logs" Functions}\label{sec:num-opt}

The update steps for $\eta$ and $\gamma$ in the EM algorithm require finding the maximum of a function
that is the following sum of log terms:
\begin{equation}\label{eqn:sum-logs}
f(x) ~~=~~ c_0\ln(x) + \sum_{i=1}^Kc_i\ln(1-w_ix)~~.
\end{equation}

This {\em sum-of-logs} function has a single parameter $x$ and a series
of $2K+1$ positive arguments: $\{c_i\}_{i=0}^K$ and  $\{w_i\}_{i=1}^K$. We assume that the weight arguments $\{w_i\}$ are
distinct and denote by $w_{\min}$ and $w_{\max}$ the minimum and maximum weights (if $w_{\min}=w_{\max}$, then $K=1$).
Note that $f(x)$ is defined in the open interval $(0,\frac 1{w_{\max}})$, and is a concave function in that interval
(since $\ln(a+bx)$ is a concave function for any choice of $a$ and $b$ and the sum of concave functions is concave as well).
Therefore, it has a unique local maximum within the interval $(0,\frac 1{w_{\max}})$, which could be found by any standard
greedy method for optimization.
The sum-of-logs function has well-defined derivatives that are simple to compute and can be used to aid the optimization
procedure. The $n$th derivative of the sum-of-logs function is defined as follows (for $n\geq 1$):
\begin{equation}\label{eqn:sum-logs-deriv}
f^{(n)} ~~=~~ -(n-1)!\left( c_0\left(-\frac{1}{x}\right)^n +  \sum_{i=1}^Kc_i\left(\frac{w_i}{1-w_ix}\right)^n\right)~~.
\end{equation}

In our implementation, we find the maximum of the function $f$ by finding the root of $f'$ without
directly evaluating the function $f$ itself. This approach expedites optimization, since the derivatives of $f$
take less time to compute than $f$ itself due to the overhead required for evaluating the $\log$ function. 
The optimization procedure is further expedited by using upper and lower bounds 
on the root of $f'$, as described below.
\begin{lemma}\label{lem:opt-bounds}
Let $f()$ be the sum-of-logs function specified in Equation \ref{eqn:sum-logs},
and let $w_{\min}$ and $w_{\max}$ denote the minimum and maximum weights in $\{w_i\}_{i=1}^K$.
Then the unique root
of $f'$ lies within the interval $[l,u]$, where 
\begin{equation}\label{eqn:opt-bounds}
l~=~\frac{c_0}{c_0w_{\max}+\sum_{i=1}^K c_iw_i}~~~~~,~~~~~u~=~\frac{c_0}{c_0w_{\min}+\sum_{i=1}^K c_iw_i} ~~.
\end{equation}
\end{lemma}
\begin{proof}
The lemma is proven by showing that $f'(l)\geq 0$ and $f'(u)\leq 0$. Let us denote $A=\sum_{i=1}^K c_iw_i$. Then, 
\begin{small}
\begin{align*}
f'(l) ~~&=~~ \frac{c_0}{l} -  \sum_{i=1}^K\frac{c_iw_i}{1-w_il} ~~=~~
c_0w_{\max}+A ~-~ \sum_{i=1}^K\frac{c_iw_i}{1-w_il}\\
&\geq~~ c_0w_{\max}+A ~-~ \frac{A}{1-w_{\max}l}~~=~~
c_0w_{\max}+A ~-~ \frac{A}{1-\frac{c_0w_{\max}}{c_0w_{\max}+A}}\\
&=~~ c_0w_{\max}+A ~-~ \frac{A}{\frac{A}{c_0w_{\max}+A}} ~~=~~
c_0w_{\max}+A ~-~ (c_0w_{\max}+A)\\
&=~~0~~.
\end{align*}
\end{small}

\vspace{-5ex}
\noindent
Similarly,
\begin{small}
\begin{align*}
f'(u) ~~&=~~ \frac{c_0}{u} -  \sum_{i=1}^K\frac{c_iw_i}{1-w_iu} ~~=~~
c_0w_{\min}+A ~-~ \sum_{i=1}^K\frac{c_iw_i}{1-w_iu}\\
&\leq~~ c_0w_{\min}+A ~-~ \frac{A}{1-w_{\min}u}~~=~~
c_0w_{\min}+A ~-~ \frac{A}{1-\frac{c_0w_{\min}}{c_0w_{\min}+A}}\\
&=~~ c_0w_{\max}+A ~-~ \frac{A}{\frac{A}{c_0w_{\min}+A}} ~~=~~
c_0w_{\min}+A ~-~ (c_0w_{\min}+A)\\
&=~~0~~.
\end{align*}
\end{small}
\end{proof}

\subsection{Dealing with Missing Data}

The probabilistic nature of our method makes it fairly easy to address the issue of missing sequence data.
Missing genotypes in the outgroup alignments are dealt with in the phylogenetic model fitting stage by
masking them with `N's in the standard way. Missing data in the individual genomes sampled
from the target population could be accommodated using two different approaches.
If missing data is sufficiently sparse, it is reasonable to discard sites with missing genotypes.
This is the approach we took in our data analysis, since the Complete Genomics individual genomes have
high confidence genotypes for $\sim$90\% of the human reference genome, and more than 75\% of the reference genome
is covered by high confidence genotypes in {\em all} 54 individuals.
However, with other data sets, it might be desirable to accommodate
sites with moderate amounts of missing genotypes. The relationship between the number samples ($n$)
and the probability of observing a polymorphism at a given site is represented in our model
by the multiplicative factor $a_n$ (see Equations \ref{eqn:poly-neut}-\ref{eqn:poly-sel} in main text).
This factor, introduced by \cite{WATT75}, corresponds to the mean total branch length (scaled by population size) 
of a genealogy with $n$ terminal branches: $a_n = \sum_{k=1}^{n-1}1/k$.
Note that if the number of sampled genomes, $n$, is constant across all sites, the factor $a_n$
serves as a constant scaling factor and its value is of no real consequence
in the inference procedure (since the multiplicative factor $\theta$ is
estimated from the data). 
However, if site $i$ has a small number of samples ($m_i$) with missing genotypes, the conditional distribution
$P\left(X_i \ |\ S_i, A_i, Z_i,\  \vect\zeta\right)$ at that site can be adjusted by replacing
the factor $a_n$ with $a_{n_i} = a_{(n-m_i)}$.

This fairly simple adjustment of the model requires several straightforward modifications
in the inference procedure. First, the neutral portion of the likelihood function is adjusted
as follows (see Equation \ref{eqn:neut-likelihood}):
\begin{small}
\begin{align}\label{eqn:neut-likelihood-missing}
 \ln\ \left(\ld_F\left(\{\theta_b\},\vect{\beta}\ ;\ \vect{X}_F, \vect{O}_F, \vect{\hat\lambda}\right)\ \right)
& =~ \sum_{b=1}^B\left(c_{F_b}\left(Y_i\neq\text{M}\right) \ln(\theta_b) ~+
\sum_{i\in F_b} c_{\{i\}}(Y_i=\text{M})\ln(1-\theta_b a_{n_i})
 ~\right)\\
\nonumber
&+~
c_F\left(Y_i=\text{L}, A_{i}=X^{maj}_{i}\right)\ln(\beta_1) ~+~ c_F\left(Y_i=\text{L}, A_{i}=X^{min}_{i}\right)\ln(\beta_3)\\
\nonumber &+~ c_F\left(Y_i=\text{H}\right)\ln(\beta_2)\\
\nonumber &+~~~C~,
\end{align}
\end{small}

\vspace{-5ex}
\noindent
implying that the MLEs of the block-specific polymorphism rates, $\theta_b$, are obtained by maximizing
the following function:
%
%
\begin{align}\label{eqn:neut-likelihood-missing}
f(\theta_b) ~=~ c_{F_b}\left(Y_i\neq\text{M}\right) \ln(\theta_b) ~+~
\sum_{i\in F_b,Y_i=\text{M}} \ln(1-\theta_b a_{i})~.
\end{align}
%
%
Notice that $f(\theta_b)$ is a ``sum-of-logs" function, and can be numerically optimized using
straightforward methods, as described above.

The selection portion of the likelihood is also affected by this adjustment to missing data
in the following way (see Equation \ref{eqn:sel-liklihood1}):
\begin{align}\label{eqn:sel-liklihood-missing}
\ln[\ld(&\rho,\eta,\gamma \ ; \data,\outgroup, \vect{\hat{\zeta}_\neut})]
 =\\
\notag &c_E(S_i=\text{\sel})\ln(\rho) ~+~c_E(S_i=\text{\neut})\ln(1-\rho) ~~+\\
\notag &c_E(S_i=\text{\sel},Z_i\neq X_i^{maj})\ln(\eta) ~+~ \sum_{b\in B}c_{E_b}(S_i=\text{\sel},Z_i= X_i^{maj})\ln(1-\eta\lambda_bt) ~~+ \\
\notag &c_E(S_i=\text{\sel},Y_i=\text{L})\ln(\gamma) ~+~ 
\sum_{b\in B}\sum_{i\in E_B}c_{i}(S_i=\text{\sel},Y_i=\text{M},Z_i= X_i^{maj})\ln(1-\gamma\theta_ba_{n_i}) ~~+~~C~,
\end{align}
and the selection parameter $\gamma$ is updated in each step of the EM algorithm by maximizing the following function
of the expected counts:
\begin{small}
\begin{eqnarray}
\label{eqn:gamma-max-missing} f_2(\gamma) &=& \langle\ c_E(S_i=\text{\sel},Y_i=\text{L})\ \rangle\ln(\gamma) ~+~ 
\sum_{b}\sum_{i\in E_b}\langle\ c_{E_b}(S_i=\text{\sel},Y_i=\text{M},Z_i= X_i^{maj})\ \rangle\ln(1-\gamma\hat\theta_ba_{n_i})~.~~
\end{eqnarray}
\end{small}

\vspace{-5ex}
\noindent
As in the case without missing data (see Equation \ref{eqn:gamma-max}), this is a ``sum-of-logs"
function that can be easily be maximized using standard numerical optimization techniques. However,
the function in the case of missing data potentially has more terms (depending on how many unique values of $n_i$
there are in each block), which would make the optimization slower. Nonetheless, small amounts of missing data
should result in no more than a moderate increase in running time.

\subsection{Estimating Approximate Standard Errors Using the Curvature Method}
\label{supp:fisherCI}

We implemented a method that uses the curvature of the likelihood function at the estimated point of
MLE in order to derive approximate standard errors for the estimates of the three selection parameters---$\rho$, $\eta$, and $\gamma$. This approach, sometimes referred to as the
``curvature method" \citep{LEHMCASE98}, produces a 3$\times$3 variance/covariance matrix for $\rho$, $\eta$, and
$\gamma$ by inverting an estimated Fisher information matrix (FIM) obtained by
negating the 3$\times$3 Hessian of the log-likelihood function (the matrix of partial second derivatives).
Formally, denote $p_1=\rho$, $p_2=\eta$, and $p_3=\gamma$, then the variance/covariance matrix
is approximated by $V=-H^{-1}$, where the Hessian, $H$, is defined as follows:
\begin{equation}\label{eqn:hess}
H= \left[\frac{\partial^2\left(\ln\left[\ld(p_1,p_2,p_3 \ ; \data,\outgroup, \vect{\hat{\zeta}_\neut})\right]\right)}{\partial p_j\partial p_k}\right]_{j,k\in \{1,2,3\}}~.
\end{equation}

The analytical computation of the Hessian at a given point is detailed later in this section.
The standard errors of $\rho$, $\eta$, and $\gamma$ are defined as the square root of the appropriate
diagonal elements of $V$. Standard errors of expected posterior counts, such as $\Ea$ and $\Ew$, are
derived by using an additional approximation.
For $\Ea$, we use use
the approximation $\Ea \approx \rho \eta \bar \lambda t$,
where $\bar \lambda$ is a weighted average of all $\lambda_b$ values (weighted according to $|E_b|$).
This approximation is derived by summing over all sites in $i\in E$ the probability, $\rho \eta \lambda_b t$, of there being a divergence under selection at site $i$ (see Table\ \ref{tab:cond-dist}). Using
a first-order Taylor approximation \citep{OEHL92}, we then estimate the variance of $\Ea$ to be:
\begin{equation}
\text{Var}[\ \Ea\ ] \approx 
(\eta\rho\bar\lambda t)^2\left(\frac{\text{Var}[\rho]}{\rho^2} + 
\frac{\text{Var}[\eta]}{\eta^2} + 
2\frac{\text{Cov}[\rho,\eta]}{\rho\eta}\right)~,
\end{equation}
where $\text{Var}[\rho]=V_{1,1}$ and $\text{Var}[\eta]=V_{2,2}$, and
$\text{Cov}[\rho,\eta]=V_{1,2}$.

The variance of $\Ew$ can be
approximated by a similar, but slightly more complex, calculation
based on the approximation
$\Ew \approx \rho \gamma a_n(\bar\theta- \eta\overline{\lambda\theta} t)$,
where $\bar \theta$ is a weighted average of all $\theta_b$ values and
$\overline{\lambda\theta}$ is a weighted average of all products
$\lambda_b\theta_b$.
Similarly to the approximation of $\Ea$, this approximation is derived by summing over all sites in
$i\in E$ the probability, $\rho \gamma(1-\eta\lambda_bt) a_n\theta_b$, of there being a polymorphism under selection at site $i$ (see Table\ \ref{tab:cond-dist}). The variance of $\Ew$
is then approximated as follows:
\begin{small}
\begin{align}
\text{Var}[\ \Ew\ ] &~\approx ~
(\rho \gamma a_n(\bar\theta- \eta\overline{\lambda\theta} t))^2 \ \times\\
&\notag \left(
\frac{\text{Var}[\rho]}{\rho^2} - 
\left(\frac{\overline{\lambda\theta}t}{\bar\theta- \eta\overline{\lambda\theta} t}\right)^2\text{Var}[\eta] + 
\frac{\text{Var}[\gamma]}{\gamma^2} - 
2\frac{\overline{\lambda\theta}t\ \text{Cov}[\rho,\eta]}{\rho(\bar\theta- \eta\overline{\lambda\theta} t)}+
2\frac{\text{Cov}[\rho,\gamma]}{\rho\gamma}-
2\frac{\overline{\lambda\theta}t\ \text{Cov}[\eta,\gamma]}{\gamma(\bar\theta- \eta\overline{\lambda\theta} t)}
\right)~.
\end{align}
\end{small}

Note that these curvature-based estimates of standard error for $\rho, \gamma, \eta, \Ea$, and $\Ew$
do not capture uncertainty in the estimates of the neutral parameters.
However, uncertainty in the neutral estimates should be fairly low assuming
a sufficient number of putative neutral sites within the relevant genomic blocks.
This can be ensured by filtering element sites in genomic blocks with too few
putative neutral sites.

\subsubsection{Analytical computation of the Hessian matrix}

We now turn to describe in detail a method for computing the Hessian matrix, $H$, for a given data set
($\vect{X},\vect{O}$) and an assignment to all model parameters: the neutral parameters
$\vect{\zeta_\neut}$, as well as the selection parameters $p_1=\rho$, $p_2=\eta$, and $p_3=\gamma$.
Due to independence across sites, the log-likelihood can be expressed as follows:
\begin{equation}
\ln\left[\ld(p_1,p_2,p_3 \ ; \data,\outgroup, \vect{\hat{\zeta}_\neut})\right] ~=~ C+
\sum_{i\in E} \ln\left[P(X_{i} \ |\ O_i,\ p_1,p_2,p_3,\ \vect{\hat{\zeta}_\neut})\right]
~,
\end{equation}
where $C$ is a term that does not depend on any of the selection parameters. Therefore, the Hessian
can similarly be expressed as a sum over sites, $H ~=~\sum_{i\in E}H^i$, where
\begin{equation}
H^i= \left[\frac{\partial^2\left(\ln\left[ P(X_{i} \ |\ O_i,\ p_1, p_2, p_3,\ \vect{\hat{\zeta}_\neut})\right]\right)}{\partial p_j\partial p_k}\right]_{j,k\in \{1,2,3\}}~.
\end{equation}

In order to compute the Hessian matrix, we thus have to compute for each site, $i\in E$,
the partial second derivatives of the site-wise likelihood function with respect to the three selection
parameters. The site-wise likelihood can be expressed using a mixture of the two selection classes:
\begin{align}
P(X_{i} \ |\ O_i,\ \rho, \eta, \gamma,\ \vect{\hat{\zeta}_\neut}) ~=~\rho~ p^{\sel}_i(\eta,\gamma)~+~ (1-\rho) p^{\neut}_i~,
\end{align}
where $p^{\sel}_i(\eta,\gamma) \equiv P(X_{i} \ |\ O_i, S_i=\sel,\  \eta,\gamma,\ \vect{\hat{\zeta}_\neut})$ is a function of $\eta$ and $\gamma$, and
$p^{\neut}_i \equiv P(X_{i} \ |\ O_i, S_i=\neut,\  \vect{\hat{\zeta}_\neut})$ is a term that does not
depend on any of the selection parameters, and is a constant that can be derived directly from
the conditional probabilities of the model. The function
$p^{\sel}_i(\eta,\gamma)$ can be expressed as follows (see Table\ \ref{tab:cond-dist}): 
\begin{small}
\begin{align}\label{eqn:sitewise-likelihood}
p^{\sel}_i(\eta,\gamma) ~=~ P(X_{i} \ |\ O_i, S_i=\sel,\ \eta,\gamma,\ \vect{\hat{\zeta}_\neut})~=~
 \begin{cases}
(1-\eta\lambda_bt)(1-\gamma\theta_ba_n)~p^{maj}_i + &\\
\frac 13 \eta\lambda_bt~(1-p^{maj}_i)& Y_i=\text{M} \\
(1-\eta\lambda_bt)\frac 13 \gamma\theta_ba_np^{maj}_i
& Y_i=\text{L}\\
0
& Y_i=\text{H}\\
\end{cases}
\end{align}
\end{small}

\vspace{-5ex}
\noindent
where $p^{maj}_i \equiv P(Z_i=X_i^{maj}\ | O_i,\hat\lambda^O_b)$.
Notice that $p^{\sel}_i(\eta,\gamma)$ has the following general form
\begin{align}\label{eqn:sitewise-likelihood}
p^{\sel}_i(\eta,\gamma) ~=~ T_i +U_i\ \eta +V_i\ \gamma+W_i\ \eta\gamma ~,
\end{align}
in which $T_i$, $U_i$, $V_i$, and $W_i$ are determined as follows:

\noindent \begin{centering}
\begin{tabular}{lllll}
\\
\hline
 & $T_i$ & $U_i$ & $V_i$ & $W_i$\\
\hline
$Y_i=\text{M}$ & 
  $p_i^{maj}$ & 
  $\frac 13\lambda_bt(1-4p_i^{maj})$&
  $-\theta_ba_np_i^{maj}$&
  $\lambda_bt\theta_ba_np_i^{maj}$\\
$Y_i=\text{L}$ &
  0 &
  0 &
  $\frac 13 \theta_ba_np_i^{maj}$ &
   $-\frac 13\lambda_bt \theta_ba_np_i^{maj}$\\
$Y_i=\text{H}$ & 0& 0& 0& 0\\
\hline
\end{tabular}
\par\end{centering}
\vspace{10pt}

Thus, by determining $T_i$, $U_i$, $V_i$, $W_i$ and $p_i^{\neut}$ for each site, the site-wise likelihood function can be re-expressed as:
\begin{align}
P(X_{i} \ |\ O_i,\ \rho, \eta, \gamma,\ \vect{\hat{\zeta}_\neut}) ~=~
\rho\cdot\left(T_i +U_i\ \eta +V_i\ \gamma+W_i\ \eta\gamma\right) ~+~ (1-\rho)\cdot p^{\neut}_i~,
\end{align}
and the site-wise Hessian matrix can be derived using the following formulas:
\begin{small}
\begin{align}
&&H^i_{1,1}&=&
\frac{\partial^2\left(\ln\left[ P(X_{i} \ |\ O_i,\ \rho,\eta,\gamma,\ \vect{\hat{\zeta}_\neut})\right]\right)}
{\partial \rho^2} &=&
-\left(
\frac{T_i +U_i\ \eta +V_i\ \gamma+W_i\ \eta\gamma-p^{\neut}_i}{P(X_{i} \ |\ O_i,\ \rho,\eta,\gamma,\ \vect{\hat{\zeta}_\neut})}\right)^2 \\
&&H^i_{2,2}&=&
\frac{\partial^2\left(\ln\left[ P(X_{i} \ |\ O_i,\ \rho,\eta,\gamma,\ \vect{\hat{\zeta}_\neut})\right]\right)}
{\partial \eta^2} &=&
-\left(
\frac{(U_i\ +W_i\ \gamma)\rho}{P(X_{i} \ |\ O_i,\ \rho,\eta,\gamma,\ \vect{\hat{\zeta}_\neut})}\right)^2 \\
&&H^i_{3,3}&=&
\frac{\partial^2\left(\ln\left[ P(X_{i} \ |\ O_i,\ \rho,\eta,\gamma,\ \vect{\hat{\zeta}_\neut})\right]\right)}
{\partial \gamma^2} &=&
-\left(
\frac{(V_i\ +W_i\ \eta)\rho}{P(X_{i} \ |\ O_i,\ \rho,\eta,\gamma,\ \vect{\hat{\zeta}_\neut})}\right)^2 \\
H^i_{1,2}&=&H^i_{2,1}&=&
\frac{\partial^2\left(\ln\left[ P(X_{i} \ |\ O_i,\ \rho,\eta,\gamma,\ \vect{\hat{\zeta}_\neut})\right]\right)}
{\partial \rho\ \partial\eta} &=&
\frac{(U_i\ +W_i\ \gamma)p_i^{maj}}{(P(X_{i} \ |\ O_i,\ \rho,\eta,\gamma,\ \vect{\hat{\zeta}_\neut}))^2} \\
H^i_{1,3}&=&H^i_{3,1}&=&
\frac{\partial^2\left(\ln\left[ P(X_{i} \ |\ O_i,\ \rho,\eta,\gamma,\ \vect{\hat{\zeta}_\neut})\right]\right)}
{\partial \rho\ \partial\gamma} &=&
\frac{(V_i\ +W_i\ \eta)p_i^{maj}}{(P(X_{i} \ |\ O_i,\ \rho,\eta,\gamma,\ \vect{\hat{\zeta}_\neut}))^2} \\%
H^i_{2,3}&=&H^i_{3,2}&=&
\frac{\partial^2\left(\ln\left[ P(X_{i} \ |\ O_i,\ \rho,\eta,\gamma,\ \vect{\hat{\zeta}_\neut})\right]\right)}
{\partial \eta\ \partial\gamma} &=&
\frac{W_ip_i^{maj}\rho(1-\rho)+\rho^2(T_iW_i-U_iV_i)}{(P(X_{i} \ |\ O_i,\ \rho,\eta,\gamma,\ \vect{\hat{\zeta}_\neut}))^2}
\end{align}
\end{small}

These site-wise Hessians, $H^i$, are summed across all sites $i\in E$, and then the Hessian is negated
and inverted to obtain the variance/covariance matrix: $V = (-\sum_{i\in E}H^i)^{-1}$.

\subsection{Computing the Posterior Expected Counts $\boldsymbol{\Ea}$ and $\boldsymbol{\Ew}$}

Given a joint assignment to all model variables, it is possible
to produce posterior expectations for various measurements that directly relate to the
the different modes of selection, namely strong positive and weak negative selection.
A useful measure for the extent to which positive selection
has affected the collection of functional elements is, $\PD$, the number of
divergences within element sites that are driven by
positive selection (also referred to as the number of adaptive divergences).
A similar measurement pertaining to
weak negative selection is $\WP$, the number of polymorphic sites subject to weak negative
selection. Expected values for $\PD$ and $\WP$ are obtained by summing
over site-wise posterior probabilities, as in the E step of the EM algorithm for selection inference:
\begin{align}
\label{eqn:E-A}
\Ea~ =&~\langle c_E(Y_i=\text{M},Z_i\neq A_i, S_i=\sel)\rangle\\
\notag = &~\sum_{i\in E\ | Y_i=\text{M}}P(Z_i\neq X_i^{maj},S_i=\sel\ | X_i,O_i,\ \vect\zeta)~,\\
\Ew~ =&~\langle  c_E(Y_i=\text{L}, S_i=\sel)\rangle\\
\notag = &~\sum_{i\in E\ | Y_i=\text{L}}P(S_i=\sel\ | X_i,O_i,\ \vect\zeta)~,
\end{align}
where $\langle c_E(\chi) \rangle$ denotes the expected number of element sites with
model variable configuration $\chi$. The site-wise posterior probabilities are computed,
as in the EM algorithm, using the joint probabilities in Table \ref{tab:em-posterior}.


These formulas makes use of our two main assumptions regarding modes of selection,
namely, that divergence at selected
sites occurs only due to positive selection, and polymorphism at selected
sites occurs only due to weak negative selection and is restricted to `L' sites.
For normalization, we will typically divide $\PD$ and $\WP$ by the total number of element sites, $|E|$
(in kilobases).
Alternatively, by normalizing $\Ea$ by the total (expected) number of divergences,
we can also obtain an estimate of the fraction of fixed
differences driven by positive selection, referred to in the literature as $\alpha$ \citep{SMITEYRE02}.
Both measures, $\Ea$ per site and $\alpha$, provide useful and somewhat complementary 
information on the extent to which positive
selection has influenced the functional elements of interest. The fraction of fixed differences, $\alpha$,
which has been used in several recent studies as a measure for positive selection \citep{SMITEYRE02,ANDO05},
describes the {\em relative} influence of positive selection on the set of
observed divergences. As such, this measure also reflects negative selection acting on the sites,
since negative selection reduces the overall number of divergences, and thus leads to an increase in $\alpha$.
Conversely, $\Ea$ per site measures the {\em absolute} influence of positive selection on the data, and as such, 
is not influenced by negative selection.

\subsection{Simulation Setup}
\label{supp:simulation}

\subsubsection{Demographic Model}

The demographic model used in simulation was designed to reflect the joint evolutionary history of humans and
their closest primate relatives: the chimpanzee, orangutan, and rhesus macaque (see Supplementary Methods).
The effective population size was held constant at $N_e = $10,000 across the outgroup portion of the
phylogeny, and divergence times of 6.5, 17.5, and 25 million years ago were assumed for 
the chimpanzee, orangutan, and rhesus macaque outgroup populations,
respectively.  These times were expressed in
generations by assuming an average generation time of 20 years throughout the phylogeny.
In order to validate the robustness of our methods to changes in ancestral population sizes,
we simulated the target population using four different demographic scenarios.
In the simplest scenario, the target population was simulated
with constant size since divergence from chimpanzee.  Another scenario contained a moderate population
expansion, and the final two scenarios contained population bottlenecks and
exponential 
expansions (Table\ \ref{tab:demography}). The intensity and timing of the bottlenecks
and expansions were taken from the demographic model suggested by \cite{GUTEETAL09},
reflecting the respective demographic histories of African, European, and East-Asian populations.

\subsubsection{Modeling recombination, mutation rate variation, and selection}

Our simulations were carried out using SFS\_CODE \citep{HERN08}, which provides a
flexible framework for
full forward simulation of sequence evolution in populations with selection.
Each simulation consisted of a synthetic block containing a 10 bp element and 5,000 neutral sites flanking it on each side. 
The synthetic blocks were simulated with a constant population-scaled
recombination rate of $\rho = 4N_er = 4.4\times 10^{-4}$ recombinations per nucleotide position,
and a population-scaled mutation rate that varied across the different simulated blocks,
sampled with a mean value of $\theta = 4N_e\mu = 7.2\times 10^{-4}$ and a
standard deviation equal to one tenth 
of the mean.
Each nucleotide position in each simulated block was assigned
to one of four selection classes: neutral evolution ($2N_es=0$),
strong negative selection ($2N_es=-100$), weak negative 
selection ($2N_es=-10$),  and positive selection ($2N_es=10$).
The 10 kb flanking sites were all assigned to the neutral class, and 
a multinomial
distribution was used to determine the number of sites in each selection
class in the 10 bp element.
Selection at WN and SN sites was applied constantly across the phylogeny.
Positive selection at SP sites was simulated in a slightly more complex way,
since the default behavior in SFS\_CODE tends to produce repeated fixation events
because positively
selected sites are always assumed to have a suboptimal allele, even after a
fixation has occurred.
SP sites were, therefore, simulated under weak negative selection ($2N_es=-10$)
across most of the population phylogeny, with a switch
to positive selection ($2N_es=10$) on the lineage leading to the target population
at the point of divergence from chimpanzee (325,000 generations ago) for a period of 310,000
generations, followed by a return to weak negative selection for the final
15,000 generations in the simulated history. This strategy provides an
opportunity for fixation of positively selected mutations, but prevents
recurrent positive selection from obscuring the signal of long term
adaptation.

\subsubsection{Technical simulation settings}

To express times in units of $2N_e$ generations, as required by SFS\_CODE,
we used the ancestral effective population size of 10,000 and assumed
a generation time of 20 years.  To save in computational cost,
we used $N_{\text{sim}} =$ 1,000 individuals in forward simulations.
Notice that as long as $N_{\text{sim}}$ is sufficiently large to limit sampling error,
this strategy should have little effect on results,
because all parameters are expressed in population-scaled form.
We used the default ``burn-in'' of $5 \times 2N_{\text{sim}}$=10,000 generations
before initiating the specified demographic scenario.
At the end of the simulations, we sampled a single haploid genome from each of the 
three outgroup populations and fifty diploid individuals (100 chromosome samples)
from the target population, closely resembling the scenario in our data analysis
(with 54 individuals).

Below is an example command line call to SFS\_CODE for a simulation of a synthetic block
that was part of the `Mix' data set used in our simulation study (Fig.\ \ref{fig:sim-results}A).
The block contains two 5 kb neutral loci flanking a 10 bp element with three neutral sites,
three SP sites, two SN sites, and two WN sites.
The mutation-scaled population size sampled for this simulation was $\theta=0.0006492911$.
\begin{quote}
sfs\_code 7 1  -o sim.mix.block11.out -n 50 -N 1000 -TE 62.5 --INF\_SITES
\textbackslash \\\hspace*{0.25in}
--theta 0.0006492911 --rho 0.00044 -L 6 5000 3 3 2 2 5000 -a N
\textbackslash \\\hspace*{0.25in}
 -W L 0 0 -W L 5 0 -W L 1 0 -W L 2 1 10 0 1 -TW 46.3 P 3 L 2 1 10 1 0
\textbackslash \\\hspace*{0.25in}
 -TW 61.75 P 3 L 2 1 10 0 1 -TW 61.75 P 4 L 2 1 10 0 1 -TW 61.75 P 5 L 2 1 10 0 1 
\textbackslash \\\hspace*{0.25in}
-W L 3 1 100 0 1 -W L 4 1 10 0 1 -TS 0 0 1 -TS 18.75 1 2 -TS 46.25 2 3 -TS 61.94 3 4 
\textbackslash \\\hspace*{0.25in}
-Td 61.95 P 4 1.23 -TS 62.15 4 5 -Td 62.15 P 5 0.17 -TS 62.448 5 6 
\textbackslash \\\hspace*{0.25in}
-Td 62.448 P 5 0.4761905 -Td 62.448 P 6 0.2428571 
\textbackslash \\\hspace*{0.25in}
-Tg 62.448 P 5 79.84042539 -Tg 62.448 P 6 109.69860461139491
\end{quote}

\subsection{Analysis of Human Genomic Elements}

\subsubsection{Variant Calling}

For individual human genome sequences we used 54 unrelated individuals taken from
the ``69 Genomes'' data set released by Complete Genomics in 2011
(\url{http://www.completegenomics.com/public-data/69-Genomes/}).
The 54 unrelated individuals were identified 
by eliminating 13 individuals from the 17-member CEPH pedigree (all but the
four grandparents) and the child in each of the two trios.
Genotype calls for these individuals were extracted from the individual ``masterVar'' files\footnote{The masterVar files are included in tar files available from
  ftp://ftp2.completegenomics.com.  The tar files currently have URLs of
  the form 
ftp://ftp2.completegenomics.com/\$GROUP/ASM\_Build37\_2.0.0/\$SAMPLE-200-37-ASM-VAR-files.tar,
where \$GROUP 
is one of {`Diversity', `Pedigree\_1463', `YRI\_trio', `PUR\_trio'} and
\$SAMPLE is the sample name.  The enclosed
masterVar files can be identified by names of the form 
masterVarBeta-\$NAME-200-37-ASM.tsv.bz2.}.
We considered variants designated as ``SNPs'' and ``length-preserving substitutions'' in the
masterVar files.  We also recorded the positions at which Complete Genomics could not
confidently assign a variant call for subsequent masking (see below).  All
other positions were assumed to be homozygous for the allele reported in
the UCSC hg19 reference genome
(Genome Reference Consortium Human Build 37).  For outgroup genomes, we used alignments from the
UCSC Genome Browser of the human reference genome (hg19) with the
chimpanzee (panTro2), orangutan (ponAbe2), and rhesus macaque (rheMac2)
genomes.  For each position in hg19, we recorded the aligned base from each
of the three nonhuman primates, or an indication that no syntenic alignment was
available at that position.

\subsubsection{Filters}

We considered the autosomes only (chr1--chr22), and applied various filters to reduce the
impact of technical errors from alignment, sequencing, genotype inference,
and genome assembly.
Our filters included repetitive sequences (simple repeats),
recent transposable elements, recent segmental duplications,
CpG site pairs, CpG islands, and regions not showing conserved synteny
with outgroup genomes. CpG site pairs (prone to hypermutability) were identified as position
pairs having a ``CG'' dinucleotide in any of the human samples or the outgroup genomes.
As a further caution, we excluded position pairs with C* in an outgroup and *G in human, to avoid
potential ancestral CpGs.  Sites in CpG islands were excluded because of
their unusual base composition and substitution patterns compared with
nearby regions.   Non-syntenic regions and gaps in the outgroup alignment were masked (by ``N"s)
individually in each outgroup genome.  
This uncertainty was later incorporated
by the phylogenetic model fitting stage of the inference (see above).
Sites with missing genotypes in one of the 54 human individual genome sequences
were masked out completely (see above Section on model adjustments for missing data).
This additional missing data filter excluded roughly 20\% of nucleotide sites in the genome.
Further details about several of these filters
are provided by \cite{GRONETAL11}.

\subsubsection{Putative Neutral Sites}

Our collection of genome-wide putative neutral sites was determined by
eliminating sites likely to be under selection.  Following a similar
procedure to that described in \cite{GRONETAL11}, we eliminated: (1) exons
of annotated protein-coding genes and the 1000 bp flanking them; (2)
conserved non-coding elements (identified by phastCons) and 100 bp flanking
them; and (3) RNA genes from GENCODE
v.11 and 1000 bp flanks. While a fraction of the remaining sites is likely to be
functional, this set should be dominated by sequence evolving under neutral
drift. Our examination of collections of short elements taken from this
``neutral'' set suggests that it contains at most a small fraction of functional
nucleotides (Fig.\ \ref{fig:genomic-results}A \&
\ref{fig:neutral-rho-hist}).

\subsubsection{Non-coding Genomic Elements from GENCODE}

We extracted several classes of non-coding genomic elements from the transcript annotations
provided by GENCODE v.13  \citep{HARRETAL06}. These annotations were
downloaded as a GTF file from
\url{http://www.gencodegenes.org/releases/13.html} and subsequently processed.
We considered ``exon" level annotations of various non-coding RNAs and
used only elements tagged as ``known" (rather than ``putative" or ``novel"). We initially
considered all four classes of noncoding RNAs: large interspersed non-coding RNAs (lincRNAs),
microRNAs (miRNAs), small nucleolar RNAs (snoRNAs), and small nuclear RNAs (snRNAs).
However, since the snRNA class had very sparse data, with 229 elements and only 2,648 sites after filtering,
we chose to remove it from the analysis. In addition to the three classes of noncoding RNAs,
we considered a class of promoter elements corresponding to 100 bp upstream the transcription
start site of known protein coding genes. Those were extracted using the ``transcript" level
entries for ``known" protein coding genes in the GTF file.

\subsubsection{Binding Sites for the GATA2 Transcription Factor}

In addition to the four GENCODE classes of noncoding elements, we analyzed a collection of binding sites for
the GATA2 transcription factor. We used binding
sites based on genome-wide chromatin
immunoprecipitation and sequencing (ChIP-seq) data from the ENCODE
project \citep{BERNETAL12}. These high-confidence binding sites were identified as part
of our study of 78 human transcription factors \citep{ARBIETAL12}. The full pipeline is
described in detail in that manuscript. Briefly, this pipeline  involved de novo motif discovery 
\citep[using MEME;][]{BAILELKA94}, manual inspection, and binding-site prediction at ChIP-seq peaks (using MAST). 
ChIP-seq data from multiple cell lines was used to predict a separate set of binding sites for each cell type,
and these sets were then merged. The sequence motif identified for GATA2,
as depicted in Fig.\ \ref{fig:genomic-results}C,
contains 11 positions, with seven positions (4-10) having information content $>\frac 12$.
In the literature, the binding sequence is often depicted by the core GATA motif (positions 5-8)
with the two flanking bases \citep{KOENGE93,MERIORKI93}.

\subsubsection{Structural Partitioning of miRNAs}

To partition each of the 1,424 primary miRNAs from GENCODE v.13
into the five structural regions shown in Fig.\ \ref{fig:genomic-results}D,
we predicted the secondary structures using the RNAfold and RNAsubopt
programs from the Vienna RNA software package \citep{HOFAETAL94,LOREETAL11}.
Human miR (mature miRNA) and miR* (star) coordinates were downloaded from miRBase rev.\ 19 
\citep{GRIFETAL06,GRIFETAL08,KOZOGRIF11}.
MiRBase does not distinguish between mature and star sequences,
so in order to label each sequence as mature or star, we examined the
total read count reported by the {\verb mature_read_count } table in the miRBase database
and selected the mature sequence as the predominantly expressed strand.
In cases where miRBase only reported the mature sequence, we inferred the star
sequence to be the complementary region in the hairpin structure.
The predicted folds of 23 transcripts contained bifurcating stems (or multi-loop structures),
and for an additional 25 transcripts, miRBase did not report a mature sequence.
These 48 transcripts were removed from the analysis, leaving 1,376 miRNAs.
Using the predicted secondary structure and the identification of the mature and
star sequences, we used custom scripts to partition the hairpin structure into the five different
components: (1) loop, (2) loop-proximal stem, (3) lower stem, (4) star, and (5) mature.

\subsubsection{Analysis of Genomic Elements Using Alternative Frequency Thresholds}

We analyzed the five classes of noncoding elements (lincRNAs, miRNAs, snoRNAs, promoters,
and GATA2 binding sites) with \ins using various thresholds  for distinguishing between
low and high frequency polymorphisms. Figure \ref{fig:genomic-thresholds} describes the
estimates of $\rho$ obtained for these five classes as a function of the frequency threshold
used. Overall, the estimates were insensitive to the chosen threshold as long as it was sufficiently
high ($f>$10\%), as indicated by our simulation study (Fig.\ \ref{fig:sim-results}B).
Estimates for GATA2 showed slightly more fluctuations than those for the other classes, possibly
due to the combined effects of sparse data and more complex patterns of polymorphism
caused by positive selection. The estimates of $\Ea$ obtained for GATA2, also
shown in Figure \ref{fig:genomic-thresholds} (bottom right), appear to follow the same fluctuation pattern as that of
the $\rho$ estimates for that class. We conclude that estimates obtained in our main analysis,
using a frequency threshold of 15\%, appear to be robust to our choice of threshold, and possibly
slightly conservative in the case of GATA2.

We also compared our estimates to those obtained by the simple site-count based estimators,
$\hat\rho_{\text{Poly}}$ and $\hat\rho_{\text{Div}}$ (red and blue horizontal lines, resp.).
As shown in our simulation study (Fig.\ \ref{fig:sim-results}), weak negative selection
(found in lincRNAs, snoRNAs, and promoters) results in under-estimation of $\rho$ by the 
polymorphism-based estimator, and positive selection in GATA2 results in under-estimation
of $\rho$ by the divergence-based estimator. We also considered a modified version of
the polymorphism-based estimates and MK-based estimates based only on high-frequency polymorphisms
\begin{align}
\label{eqn:rho-poly-h}
\hat\rho_{\text{Poly-H}} &=~~  1 - \frac{H_{E}\ |F|}{|E|\ H_{F}}~,\\
\label{eqn:mk-pd-h}
\hat\PD_{\text{-MK-H}} &=~~ 
D_E\ - \frac{H_{E}\ D_{F}}{H_{F}}~.
\end{align}
where $H_E$ and $H_F$ are the numbers of high-frequency polymorphisms in element sites and flanking sites,
respectively. The approach of discarding low-frequency polymorphisms as been used as a simple means for dealing with
the effects of weak negative selection in several previous studies of selection
\citep[e.g.,][]{FAYETAL01,ZHANGLI05,CHAREYRE08}. Comparing these estimates to our model-based estimates using
the same set of frequency thresholds (Fig.\ \ref{fig:genomic-thresholds}), we found that they provided
only a partial correction for the effects of weak negative selection. This is likely because they do not consider
patterns of divergence when estimating the fraction of sites under selection, and because they have reduced power
due to a reduction in the data \citep{EYREKEIG09}.

\comment{

For this comparison we also examined a version of the polymorphism-based
estimator, $\hat\rho_{\text{Poly}}$, that considered only high-frequency polymorphisms. Although
this version corrected some of the bias in the simple estimates .. 
(Supplementary Figure\ [ADD SUPP FIG])
In the case of --------, 
these corrected polymorphism-based estimates are
very close to our model-based estimates. On the other hand,
for other classes of elements, such as ------, the correction appears to be only partially effective, as it
does not consider divergence patterns and has the basic pitfall that it does not adequately account
for variation in mutation rate and genealogical background along the genome. In addition, ignoring
low frequency polymorphisms effectively eliminates a large portion of the polymorphisms from the data,
thus reducing power significantly \citep{EYREKEIG09}.

}

\end{document}